\documentclass[11pt]{article}
\usepackage[utf8]{inputenc}
\usepackage[T1]{fontenc}
\usepackage{authblk}

\usepackage[american]{babel}
\usepackage{amsmath}
\usepackage{amsfonts}
\usepackage{amssymb}
\usepackage{graphicx}
\usepackage{breakcites}

\usepackage{here}

\usepackage{comment}
\usepackage{multirow}
\usepackage{booktabs}
\usepackage{color}
\usepackage{subcaption}
\usepackage{tabularx}

\usepackage{apxproof}
\usepackage{wrapfig}
\usepackage{thmtools, thm-restate}

\usepackage[colorlinks,      								
citecolor=black,  								
linkcolor=black,  								
urlcolor=black,   								
]{hyperref}

\usepackage{tikz}
\usetikzlibrary{shapes,snakes}
\usetikzlibrary{arrows,automata}
\usetikzlibrary{positioning}

\theoremstyle{definition}
\newtheorem{definition}{Definition}

\theoremstyle{plain}
\newtheorem{theorem}{Theorem}
\newtheorem{proposition}{Proposition}
\newtheorem{corollary}{Corollary}
\newtheorem{lemma}{Lemma}

\theoremstyle{remark}
\newtheorem{remark}{Remark}

\usepackage{makeidx}
\makeindex

\usepackage{xspace}

\usepackage{pifont}
\newcommand{\NP}{\textsf{NP}}
\newcommand{\PSPACE}{\textsf{PSPACE}}

\newcommand{\PTIME}{\textsf{P}\xspace}

\newcommand{\FPT}{\textsf{FPT}}
\newcommand{\We}{\textsf{W[1]}}
\newcommand{\cerny}{{\v{C}}ern{\'{y}}}

\newcommand{\ore}[1][w]{\ensuremath{\propto^{l<l}_{#1 @s}}}
\newcommand{\orz}[1][w]{\ensuremath{\propto^{l \le l}_{#1 @s}}}

\newcommand{\orep}[1][w]{\ensuremath{\propto^{l<l}_{#1 @p}}}
\newcommand{\orzp}[1][w]{\ensuremath{\propto^{l \le l}_{#1 @p}}}
\newcommand{\ordp}[1][w]{\ensuremath{\propto^{l<f}_{#1 @p}}}

\DeclareMathOperator{\first}{first}
\DeclareMathOperator{\last}{last}
\DeclareMathOperator{\explore}{\texttt{explore}}
\excludecomment{cancel}

\begin{document}
	
\title{Synchronization under Dynamic Constraints} 

\author[]{Petra Wolf\thanks{The author was supported by DFG-funded project FE560/9-1}}
\affil[]{Universit\"at Trier, Germany,\\ \texttt{wolfp@informatik.uni-trier.de}}	
\date{}
\maketitle

\begin{abstract}
We introduce a new natural variant of the synchronization problem. Our aim is to model different constraints on the order in which a potential synchronizing word might traverse through the states. We discuss how a word can induce a state-order and examine the computational complexity of different variants of the problem whether an automaton can be synchronized with a word of which the induced order agrees with a given relation.
	While most of the problems are \PSPACE-complete we also observe \NP-complete variants and variants solvable in polynomial time. One of them is the careful synchronization problem for partial weakly acyclic automata (which are partial automata whose states can be ordered such that no transition leads to a smaller state),
	which is shown to be solvable in time $\mathcal{O}(k^2 n^2)$ where $n$ is the size of the state set and $k$ is the alphabet-size. The algorithm even computes a synchronizing word as a witness.
	This is quite surprising as the careful synchronization problem uses to be a hard problem for most classes of automata. We will also observe a drop in the complexity if we track the orders of states on several paths simultaneously instead of tracking the set of active states. Further, we give upper bounds on the length of a synchronizing word depending on the size of the input relation and show that (despite the partiality) the bound of the \cerny\ conjecture also holds for partial weakly acyclic automata.
\end{abstract}
\section{Introduction}
\index{dynamic constraints|(}
\index{sync under order}
\index{semi-automaton}
We call $A = (Q, \Sigma, \delta)$ a deterministic partial (semi-) automaton (DPA) if $Q$ is a finite set of states, $\Sigma$ is a finite alphabet, and $\delta \colon Q \times \Sigma \to Q$ is a (potentially partial) transition function. If $\delta$ is defined for every element in $Q \times \Sigma$, we call $A$ a deterministic complete (semi-) automaton (DCA). Clearly, every DCA is also a DPA. We do not specify any start and final states as we are only interested in the transition of states. A DCA~$A = (Q, \Sigma, \delta)$ is \emph{synchronizing} if there exists a word $w \in \Sigma^*$ such that $w$ takes every state to the same state. In that case, we call $w$ a \emph{synchronizing word} for $A$. If we are only interested in synchronizing a subset of states $S \subseteq Q$ we refer to the problem as \emph{subset synchronization}.
\index{synchronizing automaton}

\index{assembly line}
\index{parts orienter}
One of the oldest applications of the intensively studied topic of synchronizing automata is the problem of designing parts orienters, which are robots or machines that get an object in an (due to a lack of expensive sensors) unknown orientation and transform it into a defined orientation~\cite{DBLP:journals/tcs/AnanichevV04}. In his pioneering work, Natarajan \cite{DBLP:conf/focs/Natarajan86} modeled the parts orienters as deterministic complete automata where a state corresponds to a possible orientation of a part and a transition of some letter $a$ from state $q$ corresponds to applying the modifier corresponding to $a$ to a part in orientation $q$. He proved that the synchronization problem is solvable in polynomial time for -- what is later called -- the class of  \emph{orientable automata}~\cite{DBLP:journals/tcs/Ryzhikov19a} if the cyclic order respected by the automaton is part of the input.
Many different classes of automata have since been studied regarding their synchronization behavior. We refer to~\cite{DBLP:conf/lata/Volkov08,beal_perrin_2016,JALC20} for an overview.
The original motivation of designing a parts orienter was revisited in~\cite{DBLP:journals/ijfcs/TurkerY15} where T{\"{u}}rker and Yenig{\"{u}}n modeled the design of an assembly line, which again brings a part from an unknown orientation into a known orientation, where different modifiers have different costs. 
What has not been considered so far is that different modifiers can have different impact on the parts and as we do not know the current orientation we might want to restrict the chronology of applied modifiers. For example, if the part is a box with a fold-out lid, turning it upside-down will cause the lid to open. In order to close the lid one might need another modifier such as a low bar which brushes the lid and closes it again. To specify that a parts orienter should deliver the box facing upward with a closed lid one needs to encode something like: ``When the box is in the state \emph{facing down}, it later needs to be in the state \emph{lid closed}''. But this does not stop us from opening the lid again, so we need to be more precise and encode: ``After \textbf{the last time} the box was in the state \emph{facing down}, it needs to visit the state \emph{lid closed} at least once''. 
We will implement these conditions in our model of a parts orienter by enhancing a given DCA with a relation $R$. We will then consider different ways of how a synchronizing word implies an order on the states and ask whether there exists a synchronizing word whose implied state-order agrees with the input-relation $R$. 
The case-example above will be covered by the first two introduced orders. The third considered order relates to
 the following scenario:
Let us again picture the box with the lid in mind, but this time the box initially contains some water. We would like to have the box in a specific orientation with the lid open but the water should not be shed during orientating. We have a modifier that opens the lid and a modifier which rotates the box. Clearly we do not want the box to face downwards after the lid has been opened. So, we encode: ``As soon as the state \emph{lid open} has been reached, the state \emph{facing downwards} should never be entered again''. 

For every type of dynamic constraint (which we will also call \emph{order}), we investigate the computational complexity of the problem whether a given automaton admits a synchronizing word that transitions the states of the automaton in an order that is conform with a given relation. 
Thereby, we distinguish between tracking all active states simultaneously and tracking each state individually. 
We observe different complexities for different ordering concepts and get a good understanding of which ordering constraints yield tractable synchronization problems and which do not.
The complexity of the problem also depends on how detailed we describe the allowed sequence of states. 
\begin{cancel}
For example, if we provide a strict and total order on the states the problem concerning the third order can be solved in polynomial time. This is quite surprising as the class of automata admitting synchronizing words under this order is equivalent to the class of synchronizable partial weakly acyclic automata and hence the synchronization problem for this class of automata can be decided in polynomial time as well. We further observe quadratic and sub-quadratic length bounds on a shortest synchronizing word concerning the third order. Both results are quite surprising since for most classes of partial automata the problem of synchronizing the automaton without using an undefined transition is \PSPACE-complete and shortest synchronization words can have exponential lower bounds.
\end{cancel}
\section{Related Work}

The problem of checking whether a synchronizing word exists for a given DCA $A=(Q, \Sigma, \delta)$ can be solved in time $\mathcal{O}(|Q|^2|\Sigma|)$, when no synchronizing word is computed, and in time~$\mathcal{O}(|Q|^3)$ when a witnessing synchronizing word is demanded~\cite{DBLP:journals/siamcomp/Eppstein90,DBLP:conf/lata/Volkov08}. In comparison, if we only ask for a subset of states $S \subseteq Q$ to be synchronized, the problem becomes \PSPACE-complete for general DCAs~\cite{DBLP:conf/dagstuhl/Sandberg04}. These two problems have been investigated for several smaller classes of automata involving orders on states. 
Here, we want to mention the class of \emph{oriented} automata whose states can be arranged in a cyclic order which is preserved by all transitions. This model has been studied among others in \cite{DBLP:conf/focs/Natarajan86,DBLP:journals/siamcomp/Eppstein90,DBLP:journals/tcs/AnanichevV04,DBLP:journals/fuin/RyzhikovS18,DBLP:conf/lata/Volkov08}. If the order on the states is linear instead of cyclic, we get the class of \emph{monotone} automata which has been studied in \cite{DBLP:journals/tcs/AnanichevV04,DBLP:journals/fuin/RyzhikovS18}.
\begin{cancel}
The first considered class of this kind is the class of \emph{oriented} automata\footnote{This class is called \emph{monotone} in~\cite{DBLP:journals/siamcomp/Eppstein90} but later it is
called \emph{oriented} in~\cite{DBLP:journals/tcs/AnanichevV04} or \emph{orientable} in~\cite{DBLP:journals/fuin/RyzhikovS18}.} whose states can be arranged in a cyclic order which is preserved by all transitions~\cite{DBLP:journals/fuin/RyzhikovS18,DBLP:conf/lata/Volkov08}. In other words, there exists a cyclic ordering $q_1, q_2, \dots, q_n$ of the states such that for all states $p, q$ in the ordering and all letters $a \in \Sigma$, if $p \leq q$, then $\delta(p, a) \leq \delta(q, a)$. An automaton is called \emph{monotone} if the just described order is linear instead of cyclic. It is shown in~\cite{DBLP:conf/focs/Natarajan86} and~\cite{DBLP:journals/siamcomp/Eppstein90} that both synchronization and subset synchronization can be done in polynomial time for oriented automata if the order respected by the automaton is known.
For oriented automata, a tight upper bound of $(n-1)^2$ on the length of a shortest word synchronizing an arbitrary subset of states is known~\cite{DBLP:journals/siamcomp/Eppstein90,DBLP:journals/tcs/AnanichevV04} and hence the \cerny\ conjecture~\cite{cerny1964,Cer2019} holds for oriented automata.
Since the class of monotone automata is a (proper) subclass of the class of oriented automata, the same upper bound applies for monotone automata. 
In~\cite{DBLP:journals/tcs/AnanichevV04} the length of a synchronizing word for monotone 
automata is narrowed to $n-1$.
 In~\cite{DBLP:journals/tcs/AnanichevV05} the athors showed that this bound also holds for a larger class of automata which they call \emph{generalized monotonic automata}. If one asked for a word of a given \emph{interval} rank instead, in~\cite{DBLP:journals/tcs/AnanichevV04} a quadratic upper bound is given and for an interval rank of $k=2$ it is shown that no linear upper bound exists for monotone automata.

In~\cite{DBLP:journals/fuin/RyzhikovS18} Ryzhikov and Shemyakov gave a direct algorithms that checks subset synchronization of monotone automata in time $\mathcal{O}(n^2k)$ by checking if each pair in the subset can be synchronized. Note that this condition is not sufficient for subset synchronization in general as here the problem is \PSPACE-complete~\cite{DBLP:conf/dagstuhl/Sandberg04}. They also showed that a shortest word synchronizing a monotone automata can be computed in time $\mathcal{O}(n^4k)$ while the {\sc Short-Synchronizing-Word} problem
is \NP-complete for DCAs in general~\cite{Rys80,DBLP:journals/siamcomp/Eppstein90}.
It is shown in~\cite{DBLP:journals/fuin/RyzhikovS18} that finding a synchronizable subset of maximum size can be done in polynomial time for monotone automata while the decision variant of the problem is \PSPACE-complete for general binary DCAs~\cite{DBLP:journals/tcs/Ryzhikov19a}.
\end{cancel}
%
An automaton is called \emph{aperiodic}~\cite{beal_perrin_2016} if there is a non-negative integer $k$ such that for any word $w$ and any state $q$ it holds that $\delta(q, w^k) = \delta(q, w^{k+1})$. 
An automaton is called \emph{weakly acyclic}~\cite{DBLP:journals/tcs/Ryzhikov19a} if there exists an ordering of the states $q_1, q_2, \dots, q_n$ such that if $\delta(q_i, a) = q_j$ for some letter $a \in \Sigma$, then $i \leq j$.
\index{weakly acyclic automaton}
\index{WAA}
In other words, all cycles in a WAA are self-loops.
In Section~\ref{sec:results} we will consider partial WAAs.
The class of WAAs forms a proper subclass of the class of aperiodic automata.
Each synchronizing aperiodic automaton admits a synchronizing word of length at most ${n(n-1)}/{2}$~\cite{DBLP:journals/dmtcs/Trahtman07}, whereas synchronizing WAAs admit synchronizing words of linear lengths~\cite{DBLP:journals/tcs/Ryzhikov19a}.
Asking whether an aperiodic automaton admits a synchronizing word of length at most $k$ is an \NP-complete task~\cite{DBLP:conf/lata/Volkov08} as it is for general DCAs~\cite{Rys80,DBLP:journals/siamcomp/Eppstein90}. 
The subset synchronization problem for WAAs, and hence for aperiodic automata, is \NP-complete~\cite{DBLP:journals/tcs/Ryzhikov19a}. 

Going from complete automata to partial automata normally brings a jump in complexity. For example, the so called \emph{careful synchronization} problem for DPAs asks for synchronizing a partial automata such that the synchronizing word $w$ is defined on all states.
The problem is \PSPACE-complete for DPAs with a binary alphabet~\cite{DBLP:journals/mst/Martyugin14}. It is even \PSPACE-complete for DPAs with a binary alphabet if $\delta$ is undefined for only one pair in $Q\times\Sigma$~\cite{DBLP:conf/wia/Martyugin12}.
The length of a shortest carefully synchronizing word  $c(n)$, for a DPA with $|Q| = n$, differs with $\Omega(3^{\frac{n}{3}}) \leq c(n) \leq \mathcal{O}(4^{\frac{n}{3}}\cdot n^2)$~\cite{DBLP:conf/wia/Martyugin12} significantly from the cubic upper-bound for complete automata. Also for the smaller class of monotone partial automata with an unbounded alphabet size, an exponential lower bound on the length of a shortest carefully synchronizing word is known, while for fixed alphabet sizes of 2 and 3 only a polynomial lower bound is obtained~\cite{DBLP:journals/fuin/RyzhikovS18}.
The careful synchronization problem is \NP-hard for partial monotone automata
over a four-letter alphabet \cite{DBLP:journals/ijfcs/TurkerY15,DBLP:journals/fuin/RyzhikovS18}. It is also \NP-hard for aperiodic partial automata over a three-letter alphabet~\cite{DBLP:journals/tcs/Ryzhikov19a}.
In contrast we show in Section~\ref{sec:results} that the careful synchronization problem is decidable in polynomial time for partial WAAs.

In~\cite{DBLP:journals/tcs/Ryzhikov19a,DBLP:journals/fuin/RyzhikovS18} several hardness and inapproximability results are obtained for WAAs, which can be transferred into our setting as depicted in Section~\ref{sec:results}. 
We will also observe \We-hardness results from the reductions given in~\cite{DBLP:journals/tcs/Ryzhikov19a}. So far, only little is known (see for example~\cite{DBLP:journals/jcss/FernauHV15,DBLP:journals/dmtcs/VorelR15,DBLP:conf/ncma/BruchertseiferF19}) about the parameterized complexity of all the different synchronization variants considered in the literature. 


While synchronizing an automaton under a given order, the set of available (or allowed) transitions per state may depend on the previously visited states on all paths. 
This dynamic 
of allowed transitions of a state depending on the history of chosen transition 
can also be observed in weighted and timed automata \cite{DBLP:conf/fsttcs/0001JLMS14}. 
%
More static constraints given by a second automaton have been discussed in~\cite{DBLP:conf/mfcs/FernauGHHVW19}.

\section{Problem Definitions}
\label{sec:prelim}

\begin{cancel}
A deterministic automaton $A=(Q, \Sigma, \delta)$ that might either be partial or complete is simple called an \emph{automaton}. 
The transition function $\delta$ is generalized to words in the usual way. It is further generalized to sets of states $S \subseteq Q$ as $\delta(S, w) := \{\delta(q, w)\mid q \in S\}$. We will also refer to the set $\delta(S, w)$ as $S.w$. We call a state $q$ \emph{active} regarding a word $w$ if $q \in Q.w$.
We denote by $|S|$ the size of the set $S$.
With $[i..j]$ we refer to the set $\{k \in \mathbb{N}\mid i \leq k \leq j\}$.
For a word $w$ over some alphabet $\Sigma$, we denote by $|w|$ the length of $w$, by $w[i]$ the $i^ \text{th}$ symbol of $w$ (or the empty word $\epsilon$ if $i = 0$) and by $w[i..j]$ we denote the factor of $w$ from symbol $i$ to symbol $j$. 

For an automaton $A = (Q, \Sigma, \delta)$ and a word $w \in \Sigma^*$, we call $w$ a \emph{synchronizing word} for $A$ if $\delta(q, w)$ is defined for every state in $Q$ and $|Q.w| = 1$. If there exists a synchronizing word for $A$, we call $A$ \emph{synchronizing}.

We expect the reader to be familiar with the basic concepts in the fields of complexity theory, approximation theory and parameterized complexity theory. We refer to the textbooks~\cite{DBLP:books/sp/CyganFKLMPPS15,DBLP:books/daglib/0086373,DBLP:books/lib/Ausiello99} as a reference.
\end{cancel}
A deterministic semi-automaton $A=(Q, \Sigma, \delta)$ that might either be partial or complete is called an \emph{automaton}.
The transition function $\delta$ is generalized to words in the usual way. It is further generalized to sets of states $S \subseteq Q$ as $\delta(S, w) := \{\delta(q, w)\mid q \in S\}$. We sometimes refer to $\delta(S, w)$ as $S.w$. We call a state $q$ \emph{active} regarding a word $w$ if $q \in Q.w$. 
If for some $w\in \Sigma^*$, $|Q.w| = 1$  we call $q\in Q.w$ a \emph{synchronizing state}.
We denote by $|S|$ the size of the set $S$.
With $[i..j]$ we refer to the set $\{k \in \mathbb{N}\mid i \leq k \leq j\}$.
For a word $w$ over some alphabet $\Sigma$, we denote by $|w|$ the length of $w$, by $w[i]$ the $i^ \text{th}$ symbol of $w$ (or the empty word $\epsilon$ if $i = 0$) and by $w[i..j]$ the factor of $w$ from symbol $i$ to symbol $j$. 
For each state~$q$, we call the sequence of active states $q.w[i]$ for $0 \leq i \leq |w|$ the \emph{path} induced by $w$ starting at~$q$.
%
%
We expect the reader to be familiar with basic concepts in 
complexity theory, approximation theory and parameterized complexity theory.
We refer to the textbooks~
\cite{DBLP:books/sp/CyganFKLMPPS15,DBLP:books/daglib/0086373,DBLP:books/lib/Ausiello99}.
 as a reference.

We are now presenting different orders $\lessdot_w$ which describe how a word traverses an automaton. We describe how a word implies each of the three presented orders. The first two orders relate the last visits of the states to each other, while the third type of order relates the first visits. We will then combine the order with an automaton $A$ and a relation $R \subseteq Q^2$ given in the input and ask whether there exists a synchronizing word for $A$ such that the implied order of the word agrees with the relation $R$. An order $\lessdot_w$ agrees with a relation $R \subseteq Q^2$ if and only if for all pairs $(p, q) \in R$ it holds that $p \lessdot_w q$, i.e., $R \subseteq\, \lessdot_w$.

For any of the below defined orders $\lessdot_w \subseteq Q \times Q$, we define the problem of \emph{synchronization under order} and \emph{subset synchronization under order} as:
\begin{definition}[\textmd{{\sc Sync-Under-$\lessdot_w$}}]
	\index{\textsc{Sync-Under-$\lessdot_w$}}
	Given a DCA $A = (Q, \Sigma, \delta)$ and a relation $R \subseteq Q^2$. Does there exist a word $w \in \Sigma^*$ such that $|Q.w| = 1$ and $R \subseteq\, \lessdot_w$?
\end{definition}
\begin{definition}[\textmd{{\sc Subset-Sync-Under-$\lessdot_w$}}]
	\index{\textsc{Subset-Sync-Under-$\lessdot_w$}}
	Given a DCA $A = (Q, \Sigma, \delta)$, $S \subseteq Q$, and a relation $R \subseteq Q^2$. Is there a word $w \in \Sigma^*$ with $|S.w| = 1$ and $R \subseteq \lessdot_w$?
\end{definition}
It is reasonable to distinguish whether the order should include the initial configuration of the automaton or if it should only describe the consequences of the chosen transitions. 
In the former case, we refer to the problem as {\sc Sync-Under-$\mathit{0}$-$\lessdot_w$} (starting at $w[0]$), in the latter case as {\sc Sync-Under-$\mathit{1}$-$\lessdot_w$} (starting at $w[1]$), and if the result holds for both variants, we simply refer to is as {\sc Sync-Under-$\lessdot_w$}.
Examples for positive and negative instances of the problem synchronization under order for some discussed variants are illustrated in Figure~\ref{fig:expl}. 
Let $\first(q, w, S)$ be the function returning the
minimum of positions at which the state $q$ appears as an active state over all paths induced by $w$ starting at some state in $S$. Accordingly, let $\last(q, w, S)$ return the maximum of those positions.
 Note that $\first(q, w, S) = 0$ for all states $q \in S$ and $> 0$ for $q \in Q \backslash S$. If~$q$ does not
 appear on a path induced by $w$ on $S$,
then set $\first(q, w, S) := |w|+1$ and $\last(q, w, S) := -1$. In the {\sc Sync-Under-$\mathit{1}$-$\lessdot_w$} problem variant, the occurrence of a state at position 0 is ignored (i.e., if $q$ occurs only at position 0 while reading $w$ on $S$, then $\last(q, w, S) = -1$). In the following definitions let $A = (Q, \Sigma, \delta)$ be a DCA and let $p, q \in Q$. The following relations $\lessdot_w$ are defined for every word $w \in \Sigma^*$.
\begin{definition}[\textmd{Order $l<l$ on sets}]
	\index{order $l < l$}
	\label{def:o1}
	$
	p \ore q :\Leftrightarrow
	\last(p, w, Q) < \last(q, w, Q)
$.
\end{definition}


\begin{definition}[\textmd{Order $l\leq l$ on sets}]
	\index{order $l \leq l$}
	\label{def:o2}
	$
	 p \orz q :\Leftrightarrow 
	\last(p, w, Q) \leq \last(q, w, Q)$.
\end{definition}
The second order differs from the first in the sense that $q$ does not have to appear finally without~$p$, instead they can disappear simultaneously. Further, note that in comparison with order $\ore$, for a pair $(p, q)$ in order $\orz$ it is not demanded that $q$ is active after reading $w$ up to some position $i>0$. This will make a difference when we later consider the orders on isolated paths rather than on the transition of the whole state set. 
It can easily be verified that for any word $w \in \Sigma^*$ and any automaton $A=(Q, \Sigma, \delta)$ the order~$\ore$ is a proper subset of $\orz$.
For the order $\orz$, it makes no difference whether we take the initial configuration into account since states can disappear simultaneously.

So far, we only introduced orders which consider the set of active states as a whole. It did not matter which active state belongs to which path and a state on a path $\tau$ could stand in a relation with a state on some other path $\rho$. But, in most scenarios the fact that we start with the active state set $Q$ only models the lack of knowledge about the \emph{actual} current state. In practice only one state $q$ is active and hence any constraints on the ordering of transitioned states should apply to the path starting at $q$. Therefore, we are introducing variants of order~1 and 2 which are defined on paths rather than on series of state sets.

\begin{definition} [\textmd{Order $l<l$ on paths}]
	$$p \orep q :\Leftrightarrow 
	\forall r \in Q \colon 
	\last(p, w, \{r\}) < \last(q, w, \{r\}).$$
\end{definition}

\begin{definition} [\textmd{Order $l \leq l$ on paths}]
$$p \orzp q :\Leftrightarrow 
	\forall r \in Q \colon 
	\last(p, w, \{r\}) \leq \last(q, w, \{r\}).$$
\end{definition}
The orders $\orep$ and $\orzp$ significantly differ since the synchronization problem (starting at position $1$) for $\orep$ is in \NP\ while it is \PSPACE-complete for $\orzp$.
%

While the previously defined orders are bringing ``positive'' constraints to the future transitions of a word, in the sense that the visit of a state $p$ will demand for a later visit of the state $q$ (as opening the lid demands closing the lid later in our introductory example), we will now introduce an order which yields ``negative'' constraints. The third kind of order demands for a pair of states $(p, q)$ that the (first) visit of the state $q$ forbids any future visits of the state $p$ (like do not turn the box after opening the lid). This stands in contrast to the previous orders where we could made up for a ``forbidden'' visit of the state $p$ by visiting $q$ again.
The order $l<f$ will only be considered on paths since when we consider the state set~$Q$, every pair in $R$ would already be violated in position~0.
%
\begin{definition}[\textmd{Order $l<f$ on paths}]
	\index{order $l<f$}
	\label{def:o3}
$$p \ordp q :\Leftrightarrow \forall r \in Q:
	\last(p, w, \{r\}) < \first(q, w, \{r\}).
	$$
\end{definition}
Note that $\ordp$ is not transitive; e.g., for $R=\{(1,2), (2,3)\}$ we are allowed to go from 3 to~1 as long as we have not transitioned from 1 to 2 yet.
For the order $l<f$, we will also consider the special case of $R$ being a strict total order (irreflexive, asymmetric, transitive, and total).
\begin{definition}[\textmd{{\sc Sync-Under-Total-$\ordp$}}]
	\index{\textsc{Sync-Under-Total-$\ordp$}}
	\label{def:O3total}
	Given a DCA $A = (Q, \Sigma, \delta)$, a strict and total order $R \subseteq Q^2$. 
	Is there a word $w \in \Sigma^*$ with $|Q.w| = 1$ and $R \subseteq\, \ordp$?
\end{definition}
%
%
The orders on path could also be stated as LTL formulas of some kind which need to be satisfied on every path induced by a synchronizing word $w$ and our hardness results transfer to the more general problem whether a given DCA can be synchronized by a word such that every path induced by $w$ satisfies a given LTL formula. The orders on sets could be translated into LTL formulas which need to be satisfied on the path in the powerset-automaton starting in the state representing $Q$.
\index{LTL}

Using the temporal operators globally $\square$, finally $\lozenge$, and until $\mathcal{U}$, we can express the orders $\propto^{l \leq l}$, $\propto^{l < l}$, and $\propto^{l < f}$ as follows, see~\cite{DBLP:conf/podc/MannaP89,DBLP:conf/icalp/ChangMP92} for details on these operators. 
For instance, the order $p \propto^{l \leq l}q$ can be expressed as $\square (\lozenge q \vee \square \neg p)$, meaning that globally it holds that finally $q$ holds or globally $p$ does not hold; the order $p \propto^{l < l}q$ can be expressed as $\lozenge (q \wedge \square \neg p)$, meaning that finally $q$ holds and from there on $p$ does not hold anymore; and the order $p \propto^{l < f}q$ can be expressed as $\neg q\, \mathcal{U} (\square \neg p)$, meaning that there is a position from which on $p$ does not hold anymore and before that $q$ does not hold.

Despite the similarity of the chosen orders and their translated LTL formulas we need different constructions for the considered orders as the presented attempts mostly do not transfer to the other problems. 
Therefore, it is not to be expected that a general construction for restricted LTL formulas can be obtained.
Our aim is to focus on restricting the order in which states appear and disappear on a path in the automaton or on a path in the powerset-automaton (remember the introductory example).
Hence, we have chosen the stated definitions in order to investigate the complexity of problems where the LTL formula is always of the same type, i.e., comparing only the last or first appearances of states on a path. We leave it to future research to investigate other types of LTL formulas. In order to express synchronizability of Kripke structures, an extension to CTL has been introduced in~\cite{DBLP:conf/icalp/Chatterjee016}. Note that synchronization of Kripke structures is more similar to D3-directing words~\cite{DBLP:journals/actaC/ImrehS99} for unary NFAs as in contrast to general DFAs the labels of the transitions are omitted in Kripke structures.

Finally, we introduce two problems from which we will reduce from in the next section.
\begin{definition}[\textmd{{\sc Careful Sync}  (\PSPACE-complete~\cite{DBLP:journals/mst/Martyugin14})}]
	\index{\textsc{Careful Sync}}
	Given a DPA $A=(Q, \Sigma, \delta)$. Is there a word $w\in \Sigma^*$, s.t.\ $|Q.w| = 1$ and $w$ is defined on all $q\in Q$? 
\end{definition}
\begin{definition}[\textmd{{\sc Vertex Cover}  (NP-complete~\cite{DBLP:books/daglib/0086373})}]
	\index{\textsc{Vertex Cover}}
	Given a graph $G=(V, E)$ and an integer $k\leq |V|$.
	Is there a vertex cover $V'\subseteq V$ of size $|V'|\leq k$? A vertex cover is a set of states that contains at least one vertex incident to every edge.
\end{definition}
\begin{figure}[tb]
	\centering
	\scalebox{0.92}{
	\begin{tabular}{lll}
		\begin{minipage}{0.4\textwidth}
				\begin{tikzpicture}[->,>=stealth',shorten >=1pt,auto,node distance=2cm,
semithick]
\node[state] (1) {1};
\node[state] (2) [right of = 1] {2};
\node[state] (3) [below right of =2]{3};
\node[state] (4) [left of = 3] {4};
\node[state] (5) [left of = 4]{5};
\path
(1) edge [loop left] node {$a$} (1)
	edge node {$b$} (2)
(2) edge node [left] {$a$} (3)
	edge node [left]{$b$} (4)
(4) edge [bend right] node {$a$} (3)
	edge node [right]{$b$} (1)
(3) edge [bend right] node {$a$} (4)
	edge [bend right] node [right] {$b$} (2)
(5) edge node {$a$} (4)
	edge node {$b$} (1);
\end{tikzpicture}
		\end{minipage}
 &
		\begin{tabular}{c|ccccc}
			&1 & 2 & 3 & 4 & 5\\\hline
			$b$&2 & 4 & 2 & 1 & 1\\
			$a$&3 & 3 & 3 & 1 & 1\\
			$a$&4 & 4 & 4 & 1 & 1\\
			$b$&1 & 1 & 1 & 2 & 2\\
			$b$&2 & 2 & 2 & 4 & 4\\
			$a$&3 & 3 & 3 & 3 & 3\\
		\end{tabular}
	&
	\hspace{-1em}
	\renewcommand{\arraystretch}{1.2} 
	\begin{tabular}{lcc}
		 & \checkmark & \ding{55}\\
		$\ore$ & $(1, 2)$ & $(2, 4)$\\
		$\orz$ & $(2, 4)$ & $(2,1)$\\
		$\orep$ & $(1, 2)$ & $(4,5)$\\
		$\orzp$ & $(5,5)$ & $(2,4)$\\
		$\ordp$ & $(5, 2)$ & $(4,3)$\\
	\end{tabular}
	\end{tabular}
}
\caption{DCA $A$ (left) with all paths induced by $w=baabba$ (middle) and relations $R$ consisting of single pairs forming a positive, resp.~negative, instance for versions of {\sc Sync-Under-$\lessdot_w$} (right).}
\label{fig:expl}
\end{figure}

%
%
\begin{table}[tb]\centering
	\scalebox{0.87}{
	\begin{tabular}{m{0.5cm}r|cccc|ccc}
		\multicolumn{6}{c}{\emph{Synchronization}} & \multicolumn{3}{c}{\emph{Subset Synchronization}}\\
		\multicolumn{2}{l}{Order} & $l<l$ & $l\leq l$ & $l<f$ & \hspace{-0.5cm}$l$$<$$f$-tot
		&
		$l<$$/$$\leq l$ & $l<f$ & \hspace{-0.5cm}$l$$<$$f$-tot
		\\\hline
		\multirow{2}{2.5em}{Set} &
		$\mathit{0}$ & \PSPACE-c & \PSPACE-c & -- & \hspace{-0.3cm}-- &
		\PSPACE-c & -- & \hspace{-0.3cm}--\\
		&$\mathit{1}$ & \PSPACE-c & \PSPACE-c & -- & \hspace{-0.3cm}--&
		\PSPACE-c & -- & \hspace{-0.3cm}--\\\hline
		\multirow{2}{2.5em}{Path} &
		$\mathit{0}$ & in \NP & \NP-hard & \PSPACE-c & \hspace{-0.3cm}\PTIME&
		\PSPACE-c &\PSPACE-c & \hspace{-0.3cm}\NP-c\\
		&$\mathit{1}$ & in \NP & \PSPACE-c & \PSPACE-c & \hspace{-0.3cm}\NP-c &
		\PSPACE-c &\PSPACE-c &\hspace{-0.3cm}\NP-c
	\end{tabular}}
\caption{Overview of the complexity for synchronization (on the left), and subset synchronization under order (on the right) for relations $\ore$, $\orep$, $\orz$, $\orzp$, and $\ordp$ (tot.\ is short for total).}
\label{tbl:results}
\end{table}
%
%
%
%
%
\section{Main Results}
\label{sec:results}
We now investigate the complexity of the introduced problems. An overview on the obtained results is given in Table~\ref{tbl:results}. 
\begin{theorem}
	\label{thm:inPSPACE}
	For all orders $\,\lessdot\ \in \{\ore, \orep, \orz, \orzp, \ordp\}$, the problem {\sc Sync-Under-$\lessdot$} is contained in \PSPACE. Further, it is \FPT\  with parameter $|Q|$.
\end{theorem}
\begin{proof}
	Let $A=(Q, \Sigma, \delta)$ be a DCA and $R\subseteq Q^2$. 
	We decide if there exists a synchronizing word $w\in \Sigma^*$ for $A$ with $R\subseteq \lessdot_w$ for $\lessdot_w \in \{\ore, \orep, \orz,$ $ \orzp, \ordp\}$ by performing reachability tests in an enhanced version of the powerset-automaton $\mathcal{P}(A) = (\mathcal{P}(Q), \Sigma, \delta^{\mathcal{P}})$ for~$A$. 
	Therefore, we enhance $\mathcal{P}(A)$
	with the information about the set of active pairs in $R$ in every state.
	Here, a pair in $R$ is active during the transition of a word if it constrains which states might be, or need to be visited in the future. For instance, in the example in Figure~\ref{fig:expl} concerning the order $\ore$, the pair $(2,4)$ is active while reading the prefix $ba$, since the state 2 has appeared as an active state while the state 4 has not appeared without 2 as an active state yet. It is not active after reading $baa$, since now 4 is active without 2 and hence the pair $(2,4)$ is satisfied and does not demand for further state visits. The pair becomes active again after reading $baab$ since again 2 became active demanding for the state 4 to become active without 2 again.
	
	For the orders $\ore$ and $\orz$, we will enhance each state $\hat{q}\subseteq Q$ in $\mathcal{P}(A)$  with the information on the set of $\emph{active pairs}$ in $R$ in the configuration represented by $\hat{q}$. For that, the state $\hat{q}$ will be copied $2^{|R|}$ times. For the orders $\orep$, $\orzp$, $\ordp$, we enhance every state $\hat{q} \subseteq Q$ in $\mathcal{P}(A)$ by a set $S_t$ of active pairs in~$R$ \emph{for each state} $t\in \hat{q}$. Here, the state $\hat{q}$ will be copied up to $|Q|\cdot2^{|R|}$ times, and the size of the automaton $\mathcal{P}(A)$ is bounded by $2^{|Q|} |Q|2^{|R|} = 2^{\mathcal{O}(|Q|^2)}$. Hence, for every pair of start and end state the length of a shortest path connecting them is bounded by $2^{\mathcal{O}(|Q|^2)}$.
	
	First, we clarify when a pair $(p, q) \in R$ is called active in a set state $\hat{q}\subseteq Q$, respectively in a state $t\in \hat{q}$, of the automaton $\mathcal{P}(A)$ by defining the transition function $\delta^{\mathcal{P}}$: For each $\hat{q}\subseteq Q$, $\sigma \in \Sigma$, we set $E:= \{\delta(r, \sigma) \mid r\in \hat{q}\}$ and,
	for each $S\subseteq R$ we set:
	\begin{itemize}
	\item\textbf{$\ore$:} $\delta^\mathcal{P}((\hat{q}, S), \sigma) = (E, \left(S \cup \{(p, q) \in R \mid p \in E\}\right) \backslash \{(p, q)\in S \mid q \in E \wedge p \notin E \}
	)$,
	\item\textbf{$\orz$:} $\delta^\mathcal{P}((\hat{q}, S), \sigma) = (E, \left(S \cup \{(p, q) \in R \mid p \in E\}\right) \backslash \{(p, q)\in S\mid q \in E \})
	$, 
	\end{itemize}
	for each $\hat{q} = \{q_1, q_2, \dots, q_k\}$ and $S = \{(q_1,S_1), (q_2,S_2), \dots, (q_k, S_k)\}$ we set:
	\begin{itemize}
	\item \textbf{$\orep$:} $\delta^\mathcal{P}((\hat{q}, S), \sigma) = (E, \{(\delta(q_1, \sigma), S_{\delta(q_1, \sigma)}'), (\delta(q_2, \sigma), S_{\delta(q_2, \sigma)}'), \dots, $ $(\delta(q_k, \sigma), S_{\delta(q_k, \sigma)}')\})$ with $S_{\delta(q_i, \sigma)}' := \bigcup_{\{S_j \mid (q_j, S_j) \in S \wedge \delta(q_j, \sigma) = \delta(q_i, \sigma)\}}\left(S_j \cup \right.$ $\left.\{(p, q) \in R \mid p = \delta(q_i,\sigma)\} \right) \backslash \{(p, q)\in R \mid q = \delta(q_i, \sigma)\} $,
	\item \textbf{$\orzp$:} $\delta^\mathcal{P}((\hat{q}, S), \sigma) = (E, \{(\delta(q_1, \sigma), S_{\delta(q_1, \sigma)}'), (\delta(q_2, \sigma), S_{\delta(q_2, \sigma)}'), \dots, $ $(\delta(q_k, \sigma), S_{\delta(q_k, \sigma)}')\})$ with $S_{\delta(q_i, \sigma)}' := \bigcup_{\{S_j \mid (q_j, S_j) \in S \wedge \delta(q_j, \sigma) = \delta(q_i, \sigma)\}}\left(S_j \cup\right.$ $\left. \{(p, q) \in R \mid p = \delta(q_i,\sigma)\} \right) \backslash \{(p, q)\in R \mid q = \delta(q_i, \sigma)\} $,
	\item \textbf{$\ordp$:} $\delta^\mathcal{P}((\hat{q}, S), \sigma) =(E, \{(\delta(q_1, \sigma), S_{\delta(q_1, \sigma)}'), (\delta(q_2, \sigma), S_{\delta(q_2, \sigma)}'), \dots, $ $(\delta(q_k, \sigma), S_{\delta(q_k, \sigma)}')\})$ with $S_{\delta(q_i, \sigma)}' := \bigcup_{\{S_j \mid (q_j, S_j) \in S \wedge \delta(q_j, \sigma) = 
%
%
%
	 	 \delta(q_i, \sigma)\}}\left( S_j \right) \cup \{(p, q) \mid q =\delta(q_i, \sigma)\}$, if $\{(\delta(q_i, \sigma), q) \mid q \in Q\} \cap S_{\delta(q_i, \sigma)}' = \emptyset$ for all $(\delta(q_i, \sigma), S_{\delta(q_i, \sigma)}')$. Otherwise, the transition yields to the error state~$(\emptyset, \emptyset)$. 
  \end{itemize}
	Generally speaking, the transition function $\delta^\mathcal{P}$ updates the set of active states according to $\delta$ and further updates the set of active pairs $S$ according to the newly visited states. Thereby, regarding the orders 1 and 2, the visit of a new state can activate additional pairs or satisfy some pairs $(p, q)$ and hence remove them from $S$ while regarding to the third order a state visit can only activate more pairs. Here, a transition yields to the error state $(\emptyset, \emptyset)$ if it would violate any active pair.

	We set $\text{Singl} := \{(\{p\}, \emptyset)\mid p \in Q\}$.
	Depending on the order, we define different start and final states for the automaton $\mathcal{P}(A)$.:	
	\begin{itemize}
	\item \textbf{$\ore$:} Start state:  $(Q, R)$, final states: \text{Singl}.
	\item\textbf{$\orz$:} Start state: $(Q, \emptyset)$, final states: \text{Singl}.
	\item\textbf{$\orep$:} Start state: $(Q, \{(q_1, R\backslash\{(p, q)\in R \mid q = q_1\}),\dots, (q_{|Q|},$ \linebreak[5]$R\backslash\{(p, q)\in R \mid q = q_{|Q|}\})\})$ for the case including $i=0$, and $(Q, \{(q_1, R), (q_2, R), \dots, (q_{|Q|}, R)\})$ otherwise; final states:~\text{Singl}.
	\item\textbf{$\orzp$:} Start state: $(Q, \{(q_1, S_1),\dots, (q_{|Q|},S_{|Q|})\})$ with $S_i := \{(p, q) \in R \mid p = q_i\}$ for the case including $i=0$, and $(Q, \emptyset)$ otherwise; final states: \text{Singl}.
	\item\textbf{$\ordp$:} Start state: $(Q,  \{(q_1, S_1),\dots, (q_{|Q|},S_{|Q|})\})$ with $S_i := \{(p, q) \in R \mid q = q_i\}$ for the case including $i=0$, and $(Q, \emptyset)$ otherwise; final states: $(\{p\}, *)$, every state with a singleton state set and an arbitrary set of active pairs.
	\end{itemize}
	
	In each case, the automaton $A$ is synchronizable by a word $w$ with $R \subseteq \lessdot_w$ if and only if the language accepted by $\mathcal{P}(A)$ is non-empty.
	This can be checked by non-deterministically stepwise guessing a path from the start state to some final state. Since each state contains only up to $|Q|+1$ bit-strings of length up to $|Q|^2$ a state of $\mathcal{P}(A)$ can be stored in polynomial space. Hence, we can decide non-emptiness of $L(\mathcal{P}(A))$ in non-deterministic polynomial space and according to Savitch theorem~\cite{DBLP:journals/jcss/Savitch70} we can also do this using deterministic polynomial space. Furthermore, since the size of $\mathcal{P}(A)$ is bounded by $2^{\mathcal{O}(|Q|^2)}$ a recursive search for a path from the start state to any final state can be done in time $2^{\mathcal{O}(|Q|^2)}$ which gives us an \FPT\  algorithm in the parameter~$|Q|$.
\end{proof}
\label{apx:inPSPACE}
\label{sec:apxSyncINTOProof}
After giving a general \PSPACE upper bound, we now focus on lower bounds. First, we focus on the problem {\sc Sync-Under-$\orz$} and present a reduction from the \PSPACE-complete problem of {\sc Careful Sync} for DPAs. 
Since this problem is already \PSPACE-complete for binary DPAs with one undefined transition~\cite{DBLP:conf/wia/Martyugin12}, and the number of undefined transitions directly correlates to the size of the relation $R$, we get the following result:
\begin{theorem}
	\label{thm:O2-PSPACE-one-trans}
	{\sc Sync-Under-$\orz$} is \PSPACE-complete, even for $|R| = 1$ and $|\Sigma| = 2$.
\end{theorem}
\begin{proof}
%
\index{reduction from \textsc{Careful Sync}}
	We reduce from the \PSPACE-complete {\sc Careful Sync} problem for DPAs, see \cite{DBLP:conf/wia/Martyugin12,DBLP:journals/mst/Martyugin14}. Let $A = (Q, \Sigma, \delta)$ be a DPA. We construct from $A$ a DCA $A'= (Q' = Q\cup \{q_{\circleddash}, r\}, \Sigma, \delta')$ with $q_{\circleddash}, r \notin Q$. For every pair $q \in Q, \sigma \in \Sigma$ for which $\delta(q, \sigma)$ is undefined, 
	we define the transition $\delta'(q, \sigma) = q_{\circleddash}$. 
	On all other pairs $\delta'$ agrees with $\delta$.
	Further, for some arbitrary state $t\in Q$ and for all $\gamma \in \Sigma$ we set $\delta'(q_{\circleddash}, \gamma) = \delta'(t, \gamma)$ (note that this can be $q_{\circleddash}$ itself) and $\delta'(r, \gamma) = \delta'(t, \gamma)$. 
	We set the relation $R$ to $R:= \{(q_{\circleddash}, r)\}$.
	
	Assume there exists a word $w \in \Sigma^*$, $|w| = n$ that synchronizes $A$ without using an undefined transition.
	Then, $\delta(q, w[1])$ is defined for all states $q \in Q$.
	The letter $w[1]$ acts on $A'$ in the following way: 
	(1) $\delta'(r, w[1]) = \delta'(q_\circleddash, w[1]) = \delta(t, w[1])$ which is defined by assumption; 
	(2) $\delta'(Q, w[1]) \subseteq Q$ since $\delta(q, w[1])$ is defined for all states $q \in Q$. 
	The combination of (1)-(2) yields $\delta'(Q', w[1]) \subseteq Q$. We further constructed $\delta'$ such that $\delta'(Q', w[1]) = \delta(Q, w[1])$.
	Since $\delta(q, w[2..n])$ is defined by assumption for every $q \in \delta(Q, w[1])$, $\delta'$ agrees with $\delta$ on $w[2..n]$ for every $q \in \delta(Q, w[1])$. This means especially that while reading $w[2..n]$ in $A'$ on the states in $\delta'(Q', w[1])$ 
	the state $q_\circleddash$ is not reached
	and that $\delta'(Q', w) = \delta(Q, w)$. 
	Therefore, $w$ also synchronizes the automaton $A'$. 
	The state $q_\circleddash$ is 
	only active in the start configuration where no letter of $w$ is read yet and is not active anymore while reading $w$. The same holds for $r$, hence $R = \{(q_\circleddash, r)\} \subseteq\ \orz$. 
	
	For the other direction, assume there exists a word $w \in \Sigma^*$, $|w| = n$ that synchronizes~$A'$ with $(q_\circleddash, r)\in\ \orz$.
	Then, $w$ can be partitioned into $w= uv$ with $u, v \in \Sigma^*$ where $r$ is not active while reading the factor $v$ in $w$.
	The only position of $w$ in which $r$ is active due to the definition of $\delta'$ is before any letter of $w$ is read. Hence, we can set $u = \epsilon$ and $v = w$.
	As $(q_\circleddash, r)\in\ \orz$ it holds for all $i \in [1..n]$ that $q_\circleddash \notin\delta'(Q', v[1..i])$. Hence, $\delta'(q, v)$ is defined for every state $q \in Q$. Since $\delta'$ and $\delta$ agree on the definition range of~$\delta$ it follows that $v$ also synchronizes the state set $Q$ in $A$ without using an undefined transition.
\end{proof}
\begin{remark}
	\label{rm:13}
	The construction works for both variants (with and without 0) of the problem. It can further 
	be adapted for the order $\ore$ (both variants) by 
	introducing a copy $\hat{q}$ of every state in $Q \cup \{r\}$ and setting $\delta'(\hat{q}, \sigma) = q$ for every $\sigma \in \Sigma$, $q \in Q\cup \{r\}$. For all other transitions, we follow the above construction.
	We keep $R := \{(q_\circleddash, r)\}$. Since $r$ is left after $w[2]$ for any word $w\in \Sigma^*$ with $|w|\geq 2$ in order to satisfy $R$ 
	the state $q_\circleddash$
	needs to be left with $w[1]$ 
	such that afterwards $r$ is active without 
	$q_\circleddash$. Note that $q_\circleddash$ has not been copied.
\end{remark}
\begin{corollary}
	\label{cor:o1sPSPACE}
	{\sc Sync-Under-$\ore$} is \PSPACE-complete even for $|R| = 1$ and $|\Sigma| = 2$.
\end{corollary}

\begin{remark}
	The reduction presented in the proof of Theorem~\ref{thm:O2-PSPACE-one-trans} can also be applied to show the \PSPACE-completeness of {\sc Sync-Under-$\mathit{1}$-$\orzp$}. 
	Since the state~$r$ cannot be reached from any other state, the state $q_\circleddash$
	needs to be left with the first letter of any synchronizing word and must not become active again on any path. The rest of the argument follows the proof of Theorem~\ref{thm:O2-PSPACE-one-trans}. Note that the construction only works for {\sc Sync-Under-$\mathit{1}$-$\orzp$}.
	If we consider {\sc Sync-Under-$\mathit{0}$-$\orzp$} 
	the problem might become easier. But it is at least \NP-hard.
\end{remark}
\begin{cancel}
\begin{corollary}
The problem {\sc Sync-Under-$\mathit{1}$-$\orzp$} is \PSPACE-complete.
\end{corollary}
\end{cancel}
%
%
\begin{theorem}
	\label{O20PNP-hard}
	The problem {\sc Sync-Under-$\mathit{0}$-$\orzp$} is \NP-hard.
\end{theorem}
\begin{proof}
	\index{reduction from \textsc{Vertex Cover}}
	We give a reduction from {\sc Vertex Cover}. We refer to Figure~\ref{fig:VC} 
	for a schematic illustration.
	 Let $G(V, E)$ be a graph and let $k \in \mathbb{N}$. We construct from $G$ a DCA $A = (Q, \Sigma, \delta)$ in the following way. 
	We set $\Sigma = V \cup \{p\}$ for some $p \notin V$. 
	We start with $Q = \{f, r, s\}$ where~$s$ is a sink state, meaning $\delta(s, \sigma) = s$ for all $\sigma\in \Sigma$, $f$ will be the ``false way'' and $r$ will be the ``right way''. We set $\delta(r, p) = \delta(f, p) = s$ and $\delta(r, v) = r$, $\delta(f, v) = f$ for all other~$v\in \Sigma$. 
	For every edge $e_{ij} \in E$ connecting some vertices $v_i, v_j \in V$, we create two states $e_{ij}$ and $\hat{e}_{ij}$ and set $\delta(e_{ij}, v_i) = \delta(e_{ij}, v_j) = \hat{e}_{ij}$, $\delta(e_{ij}, p) = f$. 
	For all other letters, we stay in~$e_{ij}$. For the state $\hat{e}_{ij}$, we stay in~$\hat{e}_{ij}$ for all letters except~$p$. For $p$, we set $\delta(\hat{e}_{ij}, p) = s$.
	We further create 
	for $ 1 \leq i \leq k+2$ the states $q_i$ with the transitions $\delta(q_i, v) = q_{i+1}$ for $i \leq k+1$ and $v \in V$, $\delta(q_i, p) = q_i$ for $i \leq k$, and $\delta(q_{k+1}, p) = r$, $\delta(q_{k+2}, p) = s$, $\delta(q_{k+2}, v) = q_{k+2}$ for~$v \in V$. 
	We set $R := \{(q_1, r)\} \cup \{(e_{ij}, \hat{e}_{ij})\mid e_{ij}\in E\}$.
	
	If there exists a vertex cover of size $k' < k$ for $G$, then there also exists a vertex cover of size $k$ for $G$. Therefore, assume $V'$ is a vertex cover for $G$ of size $k$. Then, the word $wpp$ where $w$ is any non-repeating listing of the vertices in $V'$ is a synchronizing word for~$A$ with $R \subseteq\, \orzp[wpp]$. 
	Since $q_1$ cannot be reached from any other state, the pair $(q_1, r)\in R$ is trivially satisfied for each path starting in any state other than $q_1$. Hence, we only have to track the appearances of $q_1$ and $r$ on the path starting in $q_1$. Since $w$ lists the states in the vertex cover $V'$ it holds that $|w| = k$ and hence $q_1.w = q_{k+1}$. Further, $q_1.wp = r$ and $q_1.wpp = s$. Hence, the pair $(q_1, r)$ is satisfied on the path starting in $q_1$ as well as on all paths. 
	It remains to show that $wpp$ is indeed a synchronizing word and that all pairs in $R$ of the form $(e_{ij}, \hat{e}_{ij})$ are satisfied. 
	For every state $e_{ij}$ representing an edge $e_{ij}$, the state $\hat{e}_{ij}$ is reached if we read a letter corresponding to a vertex incident to it. Since $V'$ is a vertex cover, the word $w$ contains for each edge $e_{ij}$ at least one vertex incident to it. Hence, for each edge $e_{ij}.w = \hat{e}_{ij}$ and $e_{ij}.wpp = s$. 
	Since each state $e_{ij}$ is not reachable from any other state it follows that all pairs $(e_{ij}, \hat{e}_{ij})$ are satisfied by $wpp$ on all paths.
	It is easy to see that for all other states $q \in Q$ it holds that $q.wpp = s$.

	For the other direction, assume there exists a synchronizing word $w$ for $A$ with $R \subseteq\, \orzp $.
	By the construction of $A$ the word $w$ must contain some letters $p$. Partition $w$ into $w = upv$ where $p$ does not appear in $u$. 
	Since $R \subseteq\, \orzp$ the pair $(q_1, r)$ in $R$ enforces $|u| \leq k$ since otherwise the only path on which $q_1$ appears (namely the one starting in~$q_1$) will not contain the state $r$ as for any longer prefix $u$ it holds that $q_1.u =q_{k+2}$ and $r$ is not reachable from~$q_{k+2}$.
	The other pairs of the form $(e_{ij}, \hat{e}_{ij}) \in R$ enforces that $u$ encodes a vertex cover for $G$. Assume this is not the case, then there is some state $e_{ij}$ for which $e_{ij}.u = e_{ij}$. But then, $e_{ij}.up = f$ and from $f$ the state $\hat{e}_{ij}$ is not reachable, hence the pair $(e_{ij}, \hat{e}_{ij})$ is not satisfied on the path starting in $e_{ij}$. Therefore, $u$ encodes a vertex cover of size at most~$k$. 
\end{proof}
\begin{figure}[tb]
	\centering
	\scalebox{0.8}{
	\begin{tikzpicture}[->,>=stealth',shorten >=1pt,auto,node distance=1.8cm,
semithick,state/.style={circle, draw, minimum size=1cm}]
\node[state] (e1) {$e_{ij}$};
\node[state] (e1') [right of = e1] {$\hat{e}_{ij}$};

\node (dot1) at (0,-1) {$\vdots$};
\node (dot1') at (2,-1) {$\vdots$};
%

\node[state] (f) [below left of=e1] {$f$};
\node[state] (s) [below right of =e1'] {$s$};

\node[state] (qk2) [right of=s] {$q_{k+2}$};
\node[state] (qk1) [right of=qk2] {$q_{k+1}$};
\node[state] (qk) [right of=qk1] {$q_{k}$};
\node (dot2) [right of=qk] {$\cdots$};
\node[state] (q1) [right of=dot2] {$q_1$};
\node[state] (r) at (5.4,0) {$r$};

\path
(e1) edge node {$v_i, v_j$} (e1')
edge [loop above] node {} (e1)
edge node [above left] {$p$} (f)
(e1') edge [loop above] node {} (e1')
edge node [above right] {$p$} (s)
(f) edge [loop below] node {} (f)
edge [bend right] node {$p$} (s)
(s) edge [loop below] node {} (s)
(q1) edge [loop below] node {$p$} (q1)
(q1) edge node {} (dot2)
(dot2) edge node {} (qk)
(qk) edge [loop below] node {$p$} (qk)
(qk) edge node {} (qk1)
(qk1) edge [bend right] node [above] {$p$} (r)
(qk1) edge node {} (qk2)
(qk2) edge [loop below] node {} (qk2)
(qk2) edge node {$p$} (s)
(r) edge [bend right] node [above] {$p$} (s)
(r) edge [loop above] node {} (r)
;

\end{tikzpicture}}
\caption{Schematic illustration of the reduction from {\sc Vertex cover} (see Theorem~\ref{O20PNP-hard}). For each state, the transition without a label represents all letters which are not explicitly listed as an outgoing transition from that state.}
\label{fig:VC}
\end{figure}

If we consider $\orep$, the two variants of the order (with and without position $i=0$) do not differ since for a pair $(p, q)$, regardless of whether $p$ is reached, the state $q$ must be reached on every path. Hence, whenever we leave $q$ we must be able to return to it, so it does not matter if we consider starting in $q$ or not. In comparison with {\sc Sync-Under-$\mathit{1}$-$\orzp$}, the problem {\sc Sync-Under-$\orep$} is solvable in polynomial time using non-determinism.

\begin{theorem}
	\label{O1PNP}
	The problem {\sc Sync-Under-$\orep$} is in \NP.
\end{theorem}
\begin{proof}
	Recall that in the problem {\sc Sync-Under-$\orep$}, for every pair of states $(p, q) \in R$ and every state $r \in Q$, it is demanded that $q$ appears somewhere on a path induced by the sought synchronizing word $w$, 
	starting in $r$. Hence, a precondition for the existence of $w$ is that for every pair $(p_i, q_i) \in R$ the states $q_i$ must be reachable \emph{from any state} in~$Q$.
	More precisely, under the order $\orep$ only the last appearance of each state on a path is taken into account. Hence, a prohibited visit of a state can later be compensated by revisiting all related states in the correct order. 
	Thus, it is sufficient to first synchronize all pairs of states and then transition the remaining state through all related states in the demanded order.
	The next Lemma~\ref{lem:O1PP} proves this claim and shows that these properties can be checked in non-deterministic polynomial time.
%
\end{proof}
\begin{lemma}
	\label{lem:O1PP}
	Let $A= (Q, \Sigma, \delta)$ be a DCA and let $R \subseteq Q^2$. The automaton $A$ is synchronizable by a word $w \in \Sigma^*$ with $R\subseteq\, \orep$ if and only if the following holds:\\
	(1) For every pair of states $q_i, q_j \in Q$, there exists a word $w_{ij}$ such that $q_i.w_{ij} = q_j.w_{ij}$.\\
	(2) 
	For every state $r \in Q$, there exists a word $w_r$ such that if we consider $\orep$ only on the path induced by $w_r$, which starts in $r$, it holds that $p \orep[w_r] q$ for every pair $(p, q) \in R$. 
	
	Conditions (1) and (2) can be proved in polynomial time using non-determinism.
\end{lemma}
\begin{proof}
	Assume $A$ can be synchronized by a word $w \in \Sigma^*$ such that $R \subseteq\, \orep$. Then, $|Q.w|=1$ and hence also $|\{q_i, q_j\}.w| = 1$ for every $q_i, q_j \in Q$.
	Since $R \subseteq\, \orep$ the condition (2) already holds by definition on every path induced by $w$.
	
	For the other direction, assume condition (1) and (2) hold. Then, we can construct a synchronizing word $w=w_pw_r$ with $R \subseteq\, \orep$ in the following way:
	Start with $w_p = \epsilon$ and the set of active states $Q_\text{act} := Q$.\par
	\noindent\framebox{\parbox{\dimexpr\linewidth-2\fboxsep-2\fboxrule}{\textbf{Step 1:} If $|Q_\text{act}| = 1$ continue with Step 2, otherwise pick two arbitrary states $q_i, q_j \in Q_\text{act}$. Set $w_p := w_pw_{ij}$ and update $Q_\text{act}:= Q_\text{act}.w_{ij}$. Repeat this step.\par
	\textbf{Step 2:} We now have $Q_\text{act} = \{r\}$ for some state $r\in Q$. Return $w_pw_r$.\par}}
	
	The algorithm terminates after at most $n=|Q|$ repetitions of Step~1 since by condition~(1) in each iteration at least two states are merged. Obviously $w_pw_r$ is synchronizing for $A$. 
	Further, condition (2) gives us that  
	for every pair of states $(p, q) \in R$ the last appearance of $q$ is on the path $\tau$ induced by $w_r$ starting in $r$. This path appears on every path starting in any state $s \in Q$ as a suffix, because of $s.w_p = r$. 
	Since for $p \orep q$ the order $\orep$ only considers the last appearance of $q$ it follows that $R \subseteq\, \orep[w_r]$ since every pair in $R$ holds on the path $\tau$. 
	
	The words $w_{ij}$ in (1) can be found by determining reachability in the squared automaton $A \times A$ from the state $(q_i, q_j)$ to any singleton state\footnote{For more details see the algorithm in~\cite{DBLP:journals/siamcomp/Eppstein90} which solves general synchronizability of a DCA in polynomial time.}.
	The words $w_r$ in (2) can be computed in polynomial time using non-determinism.  
	Let $B := \{q \in Q \mid \exists p \in Q \colon (p, q) \in R\}$ be the set of all second components of pairs in $R$ with $m= |B|$. We guess an ordering $q_{i_1}, q_{i_2}, \dots, q_{i_m}$ of the states in $B$ corresponding to the last appearances of them on the path starting in $r$, induced by $w_r$. We compute the word $w_r$ in the following way, starting with $w_r = \epsilon$:\par
	\noindent\framebox{\parbox{\dimexpr\linewidth-2\fboxsep-2\fboxrule}{
	\textbf{Step 1:}
	Check whether $q_{i_1}$ is reachable from $r$ by some word $v$ using breadth-first search. If so, delete all states $p$ with $(p, q_{i_1}) \in R$ from $A$ and set $w_r := v$, otherwise return false.\par
	\textbf{Step 2:}
	For each $k$ with $1  \leq k < m$, check whether $q_{i_{k+1}}$ is reachable from $q_{i_k}$ by some word $v_k$ using breadth-first search. If so, delete all states $p$ with $(p, q_{i_{k+1}}) \in R$ from $A$ and set $w_r := w_rv_k$, otherwise return false.\par
	\textbf{Step 3:}
	Return $w_r$.}}
	
	If we guessed correctly, the algorithm returns a word $w_r$ that satisfies condition (2).
\end{proof}
\begin{remark}
	The \NP-hardness proof for {\sc Sync-Under-$\mathit{0}$-$\orzp$} in Theorem~\ref{O20PNP-hard}
	and the \NP-membership proof for {\sc Sync-Under-$\orep$} in Theorem~\ref{O1PNP} do not work for the respectively other problem since concerning $\orep$ the larger states need to be reached on \emph{every} path and not only on a path containing the corresponding smaller state as it is the case concerning~$\orzp$.
\end{remark}
\begin{theorem}
	\label{thm:o30pspace}
	The problem {\sc Sync-Under-$\mathit{0}$-$\ordp$} is \PSPACE-complete.
\end{theorem}
\begin{proofsketch}
	\index{reduction from \textsc{Careful Sync}}
	We reduce from {\sc Careful Sync}. As in the proof of Theorem~\ref{thm:O2-PSPACE-one-trans} we take every undefined transition $\delta(q, \sigma)$ to the new state $q_\circleddash$. We further enrich the alphabet by a letter $c$ and use $c$ to take $q_\circleddash$ into $Q$. We use the relation $R$ and extra states $r, s$ to enforce that $c$ is the first letter of any synchronizing word, and that afterwards $q_\circleddash$ is not reached~again. 
\end{proofsketch}
\begin{proof}
	We give a reduction from the problem {\sc Careful Sync}. Let $A = (Q, \Sigma, \delta)$ be a DPA with $c \notin \Sigma$. We construct from $A$ the DCA $A' = (Q', \Sigma \cup \{c\}, \delta')$ where $Q' = Q \cup \{q_\circleddash, r, s\}$ with $Q \cap \{q_\circleddash, r, s\} = \emptyset$.
	For every undefined transition $\delta(q, \sigma)$ with $q \in Q$, $\sigma \in \Sigma$ in $A$, we define the transition $\delta'(q, \sigma) = q_\circleddash$ in $A'$. 
	We set $\delta'(q, c) = q$ for $q \in Q$ and for some state $t$ in $Q$ we set $\delta(q_\circleddash, c) = \delta'(r, c) = \delta'(s, c) = t$. For all other letters $\gamma \in \Sigma$, we set $\delta'(q_\circleddash, \gamma) = q_\circleddash$ and $\delta'(r, \gamma) = \delta'(s, \gamma) = s$.
	On all other transitions $\delta'$ agrees with $\delta$.
	We set the relation $R$ to $R := \{(s, r)\} \cup \{q_\circleddash\} \times Q$.
	
	Assume there exists a word $w \in \Sigma^*$, $|w| = n$ that synchronizes $A$ without using an undefined transition, then $cw$ synchronizes $A'$ and $R \subseteq\, \ordp[cw]$.
	In the automaton $A'$ the letter $c$ transitions the state set $\{q_\circleddash, r, s\}$ into $Q$. 
	As $c$ is the identity on the states in $Q$, we have $\delta'(Q', c) = Q$. 
	Since $\delta'$ agrees with $\delta$ on all defined transitions of $\delta$ and $\delta(q, w)$ is by assumption defined for all states $q \in Q$ we have $\delta'(Q', cw) = \delta(Q, w) = \{p\}$ for some state $p \in Q$ and $\delta'(Q', cw[1..i]) \subseteq Q$ for all $i \leq |w|$.
	It remains to show that $R$ is consistent with~$\ordp[cw]$. 
	As $q_\circleddash$ is left with the prefix $c$ and 
	is not
	reached while reading~$w$ the subset $\{q_\circleddash\} \times Q$ of $R$ is fulfilled.
	The prefix $c$ also causes the states $r$ and $s$ to transition into the state $t$ (instead of $s$), and since $s$ is not reachable from $Q$ it is not the case that $s$ appears after $r$ on any path induced by $cw$.
	
	For the other direction, assume there exists a word $w$ that synchronizes $A'$ with $R \subseteq\, \ordp$. 
	As $(s,r) \in R$ the path induced by $w$ which starts in $r$ must not contain the state~$s$. Hence, the first symbol of $w$ must be the letter $c$ as otherwise $r$ transitions into~$s$. As $\delta'(Q', c) = Q$ all path labeled with $c$, starting in a state in $Q'$, end in a state in $Q$. For every state $q \in Q$,
	$(q_\circleddash, q)$
	 is contained in $R$. 
	Hence, $q_\circleddash$ must not
	appear on a paths labeled with $w$ starting in a state in $Q'$ after reading the first letter $c$ of $w$. This means that the word $w[2..|w|]$ synchronizes the state set $Q$ in the automaton $A'$ without leaving the state set $Q$ or using a transition which is undefined in $A$.
	Since $\delta'$ agrees with~$\delta$ on all defined $\delta$-transitions, $w[2..|w|]$ carefully synchronizes the state set $Q$ in $A$. 
\end{proof}
\label{apx:o30pspace}
\begin{remark}
	In the presented way, the reduction relies on taking the initial configuration at position $i=0$ into account but we can adapt the construction to prove \PSPACE-completeness of {\sc Sync-Under-$\mathit{1}$-$\ordp$} by copying every state in $Q$ and the state $r$. Denote a copy of a state $q$ with $q'$. Then, for each letter $\sigma$, we set $\delta'(q', \sigma) = q$, for any copied state including $r'$. Note that the copied states are not reachable from any state. Now, after the first transition $w[1]$ (which can be arbitrary), we have a similar situation as previously considered for $w[0]$. The state $r$ is active and forces the next letter to be the letter~$c$; all states in $Q$ are active; reading the letter $c$ will cause all states $q_\sigma$ to be left and never be reached again.
\end{remark}
\begin{cancel}
\begin{corollary}
	The problem {\sc Sync-Under-$\mathit{1}$-$\ordp$} is \PSPACE-complete.
\end{corollary}
\end{cancel}
In the above reduction from {\sc Careful Sync} the size of $R$ depends on $|Q|$.
 Hence, the question whether {\sc Sync-Under-$\ordp$} is \PSPACE-hard for $|R|=1$ is an interesting topic for further research.
We will now see that when $R$ is a strict and total order on $Q$, 
the problem of synchronizing under $\ordp$ (a.k.a.~{\sc Sync-Under-Total-$\ordp$}) 
 becomes tractable. 
\begin{theorem}
	\label{thm:O3totalNP}
	Let $A = (Q, \Sigma, \delta)$, $R$ be an instance of {\sc Sync-Under-Total-$\ordp$}. 
	A shortest synchronizing word $w$ for $A$ with $R \subseteq\, \ordp$ has length
	 $|w|\leq \frac{|Q|(|Q|-1)}{2}+1$. 
\end{theorem}
\begin{proof}
	The relation $R$ implies a unique ordering $\sigma$ of the states in $Q$. 
	We put a token in every state which will be moved by applications of letters and think of active states as states containing a token. 
	In the problem variant {\sc Sync-Under-Total-$\mathit{1}$-$\ordp$} the tokens can be moved anywhere in the first step but afterwards - and in the variant {\sc Sync-Under-Total-$\mathit{0}$-$\ordp$} - the tokens can only move to bigger states concerning~$\sigma$. Each letter should move at least one token and the tokens in the $|Q|$ states can only be moved $0, 1, 2, \ldots, |Q|-1$ times, giving a total length bound of $\frac{|Q|(|Q|-1)}{2}+1$ for {\sc Sync-Under-Total-$\mathit{1}$-$\ordp$} and $\frac{|Q|(|Q|-1)}{2}$ for {\sc Sync-Under-Total-$\mathit{0}$-$\ordp$}.
\end{proof}
\begin{cancel}
\begin{proof}
	The relation $R$ implies a unique ordering $\sigma$ of the states in $Q$. 
	We put a token in every state which will be moved by applications of letters and think of active states as states containing a token. 
	In the problem variant {\sc Sync-Under-Total-1-$\ordp$} the tokens can be moved anywhere in the first step but afterwards - and in the variant {\sc Sync-Under-Total-0-$\ordp$} - the tokens can only move to bigger states concerning $\sigma$. Each letter should move at least one token and the tokens in the $n$ states an only be moved $0, 1, 2, \ldots, n-1$ times, giving a total length bound of $\frac{|Q|(|Q|-1)}{2}+1$ for {\sc Sync-Under-Total-1-$\ordp$} and $\frac{|Q|(|Q|-1)}{2}$ for {\sc Sync-Under-Total-0-$\ordp$}.
	
	Since $R$ is total, transitive, and irreflexive, for every word $w$ that synchronizes $A$ with $R \subseteq\, \propto^3_w$ it holds that, if any path induced by $w$ contains a cycle, then this cycle must be a loop, i.e., it consists of only one distinct state. 
	Further, every path $\tau$ induced by $w$ must be a sub-sequence of the ordering $\sigma$ since otherwise some pair in $R$ would be violated. 
	As there are $n$ many paths induced by a word $w$ on $n$ states and a path starting in a state at position $i$ may contain only up to $n-i$ different state changes there are in total only up to $\frac{n(n-1)}{2}$ \emph{different} sets of active states in $A$ while reading $w$.

	If for some position $i$ in $w$ it holds that $\delta(q, w[1..i]) = \delta(q, w[1..(i+1)])$ for all states $q \in Q$, then the $i+1^\text{st}$ letter can be removed from $w$ yielding a shorter word $w'$ which also synchronizes $A$ under the order $R$.
	Repeatedly removing such positions from $w$ yields a word $w'$ of size at most $\frac{n(n-1)}{2}$. 
\end{proof}
\end{cancel}
Note that this length bound is smaller than the bound of the \cerny\ conjecture for $|Q| > 3$~\cite{cerny1964,Cer2019}. The same bound can be obtained for  {\sc Subset-Sync-Under-Total-$\ordp$}.
\begin{cancel}
\begin{corollary}
	\label{thm:O31length}
	Let $A = (Q, \Sigma, \delta)$, $R \subseteq Q^2$ be an instance of the problem {\sc Sync-Under-Total-$\mathit{1}$-$\ordp$} with $n:= |Q|$. If $A$ is synchronizable by a shortest word $w$ such that $R \subseteq\, \ordp$, then $|w|\leq n^2$.
\end{corollary}
\end{cancel}
We will now prove that the problem {\sc Sync-Under-Total-$\mathit{0}$-$\ordp$} is equivalent -- concerning polynomial time many-one-reductions (depicted by $\equiv_p$)~-- to the problem of carefully synchronizing a partial weakly acyclic automaton (PWAA)
 (a PWAA is a WAA where $\delta$ might be only partially defined). 
 The obtained length bound also holds for PWAAs, 
 which is only a quadratic increase w.r.t.~the linear length bound in the complete case~\cite{DBLP:journals/tcs/Ryzhikov19a}.
%
\begin{theorem}
	\label{O3TotalEquivPWAA}
	{\sc Sync-Under-Total-$\mathit{0}$-$\ordp$} $\equiv_p$ {\sc Careful Sync} of PWAAs.
\end{theorem}
\index{\textsc{Careful Sync} of PWAAs}
\begin{proof}
	We prove this statement by reducing the two problems to each other. Let $A=(Q,\Sigma,\delta)$, $R\subseteq Q^2$ be an instance of {\sc Sync-Under-Total-$\mathit{0}$-$\ordp$}. Since $R$ is a strict total order on $Q$, we can order the states according to $R$. We construct from $A$ the PWAA $A'=(Q, \Sigma, \delta')$ by removing all transitions in $\delta$ which are leading backwards in the order. Clearly, $A'$ is carefully synchronizable if and only if $A$ is synchronizable with respect to $R$.
	
	For the other reduction, assume $A=(Q, \Sigma, \delta)$ is a PWAA. Then, we can order the states in $Q$ such that no transition leads to a smaller state. We are constructing from $A$ the DCA $A' = (Q\cup \{q_<\}, \Sigma, \delta')$ and insert $q_<$ as the smallest state in the state ordering. Then, we define in $\delta'$ all transitions $(q, \sigma)$  for $q \in Q, \sigma \in \Sigma$ which are undefined in $\delta$ as~$\delta'(q, \sigma) = q_<$. We take the state $q_<$ with every symbol to the maximal state in the order. Note that the maximal state needs to be the synchronizing state if one exists. We set $R = \{(p,q) \mid p < q \text{ in the state ordering of } Q \text{ in } A\} \cup \{(q_<, q)\mid q \in Q \}$. 
	Every undefined transition $(p, \sigma)$ in $A$ is not allowed in $A'$ at any time, since otherwise the pair $(q_<, p) \in R$ would be violated. The state $q_<$ itself can reach the synchronizing state with any transition. Hence, $A'$ is synchronizable with respect to $R$ if and only if $A$ is carefully synchronizable.
\end{proof}
\noindent{\sffamily\textcolor{darkgray}{\textbf{Remarks on the length bound of synchronizing words for PWAAs:}}}
\index{synchronizing words for PWAAs}
\label{apx:PWAA-length}
In the reduction from \textsc{Careful Sync} of PWAAs to 
{\sc Sync-Under-Total-$\mathit{0}$-$\ordp$} the state set is only increased by one additional state $q_<$. As this state is not reachable from any other state (as otherwise the order would be violated) and is left into the largest state, w.r.t.~the constructed order, with every letter, this state does not contribute to the length of a potential synchronizing word if the number of states is >1. As for all other states, the allowed transitions in the DCA act in the same way as they do in the PWAA, the length bound of a synchronizing word for DCAs w.r.t.~{\sc Sync-Under-Total-$\mathit{0}$-$\ordp$} translates to a length bound of a carefully synchronizing word for PWAAs.
This is quite surprising as in general shortest carefully synchronizing words have an exponential lower bound~\cite{DBLP:conf/wia/Martyugin12}. Further, we show that careful synchronization for PWAAs is in \PTIME while the problem is \PSPACE-complete for general DPAs even if only one transition is undefined~\cite{DBLP:conf/wia/Martyugin12}.
\begin{corollary}
	For every PWAA $A=(Q, \Sigma, \delta)$, a shortest word $w$ carefully synchronizing $A$ has length $|w|\leq \frac{|Q|(|Q|-1)}{2}$.
\end{corollary}
We can generalize the length bound obtained in the case when $R$ is a total order on the whole state set to the case that $R$ is only total for a subset of states.
\begin{theorem}
	\label{thm:O3paraLength}
	Let $A = (Q, \Sigma, \delta)$, $R \subseteq Q^2$
	 with $n= |Q|$. Let $Q_1 \subseteq Q$ be such that $R$ restricted to $Q_1\times Q_1$ is a strict and total order. 
Let $p = |Q|-|Q_1|$. 
	For {\sc Sync-Under-$\ordp$}: If $A$ is synchronizable by a shortest word $w$ with $R \subseteq\, \ordp$, then: $|w|\leq (\frac{n(n-1)}{2}+1)\cdot 2^p$.
\end{theorem}
\begin{proof}
	As before, the states in the set $Q_1$ can be ordered according to $R$ and might only be traversed in this order. For every transition of a state in $Q_1$, in the worst case all possible combinations of active states in $Q\backslash Q_1$ might be traversed (once) yielding~$2^p$ transitions with identical active states in $Q_1$ between two transitions of any state in~$Q_1$. 
\end{proof}
We now present an $\mathcal{O}(|\Sigma|^2|Q|^2)$ algorithm for {\sc Sync-Under-Total-$\mathit{0}$-$\ordp$}. The idea is the following: First, we delete all transitions that violate the state order. Then, we start on all states as the set of active states and pick a letter, which is defined on all active states and maps at least one active state to a larger state in the order $R$. We collect the sequence $u$ of applied letters and after each step, we apply the whole sequence $u$ on the set of active states. This is possible as we already know that $u$ is defined on $Q$. We thereby ensure that a state which has become inactive after some iteration never becomes active again after an iteration step and hence $\Sigma_{\text{def}}$ grows in each step and never shrinks. While a greedy algorithm which does not store $u$ runs in $\mathcal{O}(|\Sigma||Q|^3)$, with this trick we get a running time of $\mathcal{O}(|\Sigma|^2|Q|^2)$. As in practice $|Q| \gg |\Sigma|$ this is a remarkable improvement. Note that we can store $u$ compactly by only keeping the map induced by the current $u$ and storing the sequence of letters $\sigma$ from which we can restore the value of $u$ in each iteration.  

	\noindent{\sffamily\textcolor{darkgray}{\textbf{Greedy algorithm for \textsc{Sync-Under-Total-$\mathit{0}$-$\ordp$}}}}
	\label{alg:greedy}
	Let $A=(Q, \Sigma, \delta)$ be a DCA with $|Q| = n$ and $|\Sigma|=m$, and let $R \subseteq Q \times Q$ be a total order.
	We sketch an $\mathcal{O}(mn^3)$ greedy algorithm which computes a synchronizing word $w$ for $A$ with $R \subseteq \ordp$.
	
	First, order the set $Q$ according to $R$. Delete all transitions in $A$ which are leading backwards in the state-ordering obtaining the DPA $A'$. Set $Q_1 = Q$, $w_1 = \epsilon$.
	
	At each step $i$: Check if $|Q_i| = 1$, if so return yes and the word $w_i$. Otherwise, compute $\Sigma_i = \{\sigma \in \Sigma \mid q.\sigma \text{ is defined for all } q \in Q_i\}$. Test if there is at least one letter $\sigma \in \Sigma_i$ that maps a state in $Q_i$ to a larger state. If so, apply this letter to $Q_i$, obtaining $Q_{i+1} = Q_i.\sigma$, set $w_{i+1}=w_i\sigma$, and continue with the next step. If there is no such letter $\sigma$, return no.
	
	By Theorem~\ref{O3TotalEquivPWAA} there exists a synchronizing word $w$ for $A$ with $R \subseteq \ordp$ if and only if there exists a carefully synchronizing word $w$ for $A'$.
	Observe that if $A'$ is carefully synchronizing, then every subset of $Q$ can be synchronized. Hence, if $A'$ is carefully synchronizing, then for every subset $S\subseteq Q$ there exists a letter $\sigma$ which is defined on all states in $S$ and maps at least one state in $S$ to a larger state. Hence, the algorithm will find a carefully synchronizing word and terminate. 
	
	Conversely, if  the algorithm returns no, the set $Q_i$ of active states at the last step is a witness proving that $A'$ is not carefully synchronizing, since no letter can map this subset to a different one, and thus $Q$ cannot be synchronized.
	
	The preprocessing of the algorithm takes time $\mathcal{O}(n \log n + mn)$. Each step of the algorithm takes time $\mathcal{O}(mn)$. The maximum number of steps is $\mathcal{O}(n^2)$, since at each step we move a token on the active states according to a total order by at least one. Hence, the total running time is $\mathcal{O}(mn^3)$.
\begin{theorem}
	\label{thm:O3totalP}
	 {\sc Sync-Under-Total-$\mathit{0}$-$\ordp$} is solvable in quadratic time.
\end{theorem}
\begin{proof}
	Let $A = (Q, \Sigma, \delta)$ be a DCA, and let $R \subseteq Q^2$ be a strict and total order on~$Q$.
	Figure~\ref{fig:algo} describes an algorithm that decides in time $\mathcal{O}(|\Sigma|^2|Q|^2)$ whether $A$ is synchronizable with respect to $R$ under the order $\ordp$ (including position $i=0$) on paths.
	Despite the simplicity of the algorithm its correctness is not trivial and is proven in the following lemmas.
\end{proof}
		\begin{figure}[h!]
	\noindent
	\framebox{\parbox{\dimexpr\linewidth-2\fboxsep-2\fboxrule}{
	\textbf{Step 1:} Order all states in $Q$ according to the order $R$. Since $R$ is strict and total the states can be ordered in an array $\{q_1,q_2,\dots,q_n\}$.\par
	\textbf{Step 2:} Delete in the automaton $A$ all transitions which are leading backwards in the state-ordering. If this produces a state with no outgoing arc, abort; return false. 
	\par
	\textbf{Step 3:} Let $q_n$ be the maximal state according to the order $R$. Delete all transitions in $A$ which are labeled with letters $\sigma \in \Sigma$ for which $q_n.\sigma$ is undefined.  If this produces a state with no outgoing transition, abort and return false.\par
	\textbf{Step 4:} Partition the alphabet $\Sigma$ into $\Sigma_\text{def}$, consisting of all letters $\sigma \in \Sigma$ for which $q.\sigma$ is defined for all states $q \in Q$, and $\Sigma_\text{par} := \Sigma \backslash \Sigma_\text{def}$. If $\Sigma_\text{def} = \emptyset$ abort; return~false. \par
	\textbf{Step 5:}
	Compute $\explore(Q, Q, \Sigma_\text{def}, \epsilon)$ which returns $Q_\text{act}$ and $u \in \Sigma_\text{def}^*$. 
	The returned set of active states will equate $Q_\text{trap} = \{ q \in Q \mid q.\Sigma_\text{def} = q\}$.
\par
	\textbf{Step 6:} 
	Set $\Sigma_\text{def} := \Sigma_\text{def} \cup \{\sigma \in \Sigma_\text{par} \mid q.\sigma \text{ is defined for all } q\in Q_\text{act} \}$.
	
	Compute $\explore(Q, Q_\text{act}, \Sigma_\text{def}, u)$ which returns $Q_\text{act}'$ and $u' \in \Sigma_\text{def}^*$. \par
	Set $Q_\text{act} := Q_\text{act}'$, $u := u'$, 
	$\Sigma_\text{par} := \Sigma \backslash \Sigma_\text{def}$.

	Repeat this step until $Q_\text{act}$ does not change anymore ($\equiv$ to $\Sigma_{\text{def}}$ does not change anymore). 
	
	Then, if $Q_\text{act} = \{q_n\}$ return true, otherwise return false.

	\textbf{Procedure $\explore$:} 
	Input: Ordered state set $Q$, set of active states $Q_\text{act}$, alphabet $\Sigma_\text{exp}$ to be explored, word $u$ with $Q.u = Q_\text{act}$.\par
	Initialize a new word $u' := u$.\par
	
	Go through the active states in order. For the current state $q$, test if any $\sigma \in \Sigma_\text{def}$ leads to a larger state, if so, perform the transition $\sigma u$ on all active states and update the set of active states $Q_\text{act}$. Concatenate $u'$ with $\sigma u$. 
	Continue with the next larger active state (not that this can be $q.\sigma u$). If $q_n$ is reached, return $u'$, and the current set of active states $Q_\text{act}$.
}}
\caption{Polynomial time algorithm 
	for {\sc Sync-Under-Total-$\mathit{0}$-$\ordp$} on the input\linebreak[3] $A=(Q, \Sigma, \delta)$, $R\subseteq Q^2$.}
\label{fig:algo}
	\end{figure}
%
%
%
\begin{lemma}
	\label{lem:O3totalRuntime}
	The algorithm in Figure~\ref{fig:algo} terminates on every input $A=(Q, \Sigma, \delta)$ with $m=|\Sigma|$, $n=|Q|$, strict and total order $R \subseteq |Q|^2$ in time $\mathcal{O}(m^2n^2)$.
\end{lemma}
\begin{proof}
	Step~1 can be performed in time $\mathcal{O}(n \log n)$ using the Quicksort-algorithm.
	Step~2 to Step 5 take time $\mathcal{O}(mn)$ each. The procedure $\explore$ takes time $\mathcal{O}(mn^2)$. The number of iterations in Step 6 is bounded by $|\Sigma_\text{part}|$ as $\Sigma_{\text{def}}$ is applied exhaustively on $Q_\text{act}$ and by invariant (2) of Lemma~\ref{lem:O3Pyes} we have $Q'_\text{act} \subseteq Q_\text{act}$, This yields
	a total run-time of $\mathcal{O} (m^2 n^2)$.
\end{proof}
\begin{lemma}
	\label{lem:O3Pyes}
	If the algorithm in Figure~\ref{fig:algo} returns true on the input $A=(Q, \Sigma, \delta)$, strict and total order $R \subseteq |Q|^2$, then $A$, $R$ is a yes instance of {\sc Sync-Under-Total-$\mathit{0}$-$\ordp$}.
\end{lemma}
\begin{proof}
	For the procedure $\explore$, the following invariant holds: Let $u_\text{old}$ be the word $u$ before the execution of $\explore$ and let $u_\text{new}$ be the one after the execution of $\explore$. Then, it holds for all executions of $\explore$ that (1) $Q.u_\text{new}$ is defined, (2) $Q.u_\text{new} \subseteq Q.u_\text{old}$, and (3) $Q.u_\text{new}u_\text{new} = Q.u_\text{new}$.
	We prove the invariant by induction. 
	First, note that the word $u$ computed by $\explore$ in Step 5 is defined on all states in $Q$ since it only consists of letters which are defined on all states. Since we go through the states in order during the execution of $\explore$ and we only proceed with the next larger state if (1) we where able to leave the current one towards a larger state or if (2) the current state cannot be left with any of the letters in $\Sigma_\text{def}$, it holds that $Q.uu = Q.u$. Also, trivially $Q.u \subseteq Q$.
	
	Next, consider some later execution of $\explore$.
	The new word computed by $\explore$ is of the form $u_\text{new} :=  u_\text{old} \sigma_1 u_\text{old} \sigma_2 u_\text{old} \dots \sigma_i u_\text{old}$ for some $0 \leq i \leq |Q|$. The induction hypothesis tells us that (1) $Q.u_\text{old}$ is defined. Since $Q.u_\text{old}\sigma_1$ is defined (since $\sigma_1 \in \Sigma_\text{def}$) and $Q.u_\text{old}\sigma_1 \subseteq Q$ it holds that $Q.u_\text{old}\sigma_1u_\text{old}$ is defined. Further, since $u_\text{old}$ brings all states to the set $Q.u_\text{old}$ it also brings a subset of $Q$ to a subset of $Q.u_\text{old}$. Using the induction hypothesis (3) we get by an induction on $i$ that $Q.u_\text{new}$ is defined and $Q.u_\text{new} \subseteq Q.u_\text{old}$. Since in the execution of $\explore$ we only proceed with the next larger state if we exhaustively checked all possible transitions for the current state and since $Q.u_\text{old}u_\text{old} = Q.u_\text{old}$ it follows that $Q.u_\text{new}u_\text{new} = Q.u_\text{new}$.
	
	If the algorithm in the proof of Theorem~\ref{thm:O3totalP} terminates and returns yes, it also returns a synchronizing word $u$. By the invariant proven above, we know that $Q.u$ is defined. This means that $u$ never causes a transition of a larger state to a smaller state and hence $\ordp[u]$ agrees with $R$. 
	During the execution of the algorithm we track the set of active states~$Q_\text{act}$ (starting with~$Q$) and only return true if $Q_\text{act}$ contains only the in $R$ largest state~$q_n$. Since $R$ is a total order, every $q \in Q$ is smaller than $q_n$ and hence $q_n$ cannot be left. Therefore, $q_n$ needs to be the single synchronizing state of $A$ and $u$ is a synchronizing word for~$A$.
%
%
\end{proof}
\begin{lemma}
	\label{lem:O3Pno}
	If the algorithm in Figure~\ref{fig:algo} returns false on the input $A=(Q, \Sigma, \delta)$ and a strict and total order $R \subseteq |Q|^2$, then $A$ is not synchronizable under the order~$\ordp[w]$ with respect to the input order $R$ .
\end{lemma}
\begin{proof}
	The algorithm returns false in the following cases.\par
	(1) All outgoing transitions of some state $q$ are deleted in Step 2. In that case, every transition of $q$ leads to a smaller state. As this would violate the order $R$, we cannot perform any of those transitions. Hence, $q$ cannot be left. (The case that $q = q_n$ is treated in (2).)\par
	(2) Since $q_n$ is the largest state, it cannot be left. Hence, $q_n$ will be active the whole time. Therefore, any transition which is not defined for $q_n$ cannot be taken at all since $q_n$ is active during the whole synchronizing process. Hence, we can delete these transitions globally. If this creates a state which cannot be left anymore, this state cannot be synchronized. \par
	(3) The execution of $\explore$ returns two identical sets of active states $Q_\text{act}$ in a row. Let $\Sigma_\text{def}$ be the explored alphabet of the last execution of $\explore$. Then, $\Sigma_\text{def}$ contains all letters $\sigma$ from $\Sigma$ for which $q.\sigma$ is defined on all states $q \in Q_\text{act}$ and none of them leads some state in $Q_\text{act}$ to a larger state. Since the relation $R$ forbids cycles, for all $\sigma \in \Sigma_\text{def}$ and all $q \in Q_\text{act}$ $q.\sigma = q$ and hence this set cannot be left when all states of the set are active simultaneously. 
	Since all states are active at the beginning of the algorithm, also all states in $Q_\text{act}$ are active and since this set cannot be left with any transition which does not cause an undefined transition for all states in the set, the state set cannot be synchronized at all.
\end{proof}
%
\begin{corollary}
	The careful synchronization problem for PWAA is in \PTIME.
\end{corollary}
\index{\textsc{Careful Sync} of PWAAs}
If we allow one unrestricted transition first ({\sc Sync-Under-Total-$\mathit{1}$-$\ordp$}) the problem is related to the subset synchronization problem of complete WAAs which is \NP-complete~\cite{DBLP:journals/tcs/Ryzhikov19a}. Together with the quadratic length bound of a synchronizing word of  {\sc Sync-Under-Total-$\mathit{1}$-$\ordp$} (which implies membership of {\sc Sync-Under-Total-$\mathit{1}$-$\ordp$} in \NP), we get:
\begin{theorem}
	\label{thm:to1NPc}
	The problem {\sc Sync-Under-Total-$\mathit{1}$-$\ordp$} is \NP-complete.
\end{theorem}
\index{reduction from WAA \textsc{Sync-From-Subset}}
\begin{proof}
			\label{prf:to1NPc}
			We reduce from the \NP-complete problem: Given a complete weakly acyclic automaton $A=(Q, \Sigma, \delta)$ and a subset $S \subseteq Q$, does there exist word $w \in \Sigma^*$ such that $|S.w|=1$. 
		We construct from $A$ an automaton $A'=(Q', \Sigma \cup \{c\}, \delta')$ with $c \notin \Sigma$ in the following way. 
			A schematic illustration of the construction is depicted in Figure~\ref{fig:syncsubset}.
		We start with $Q' = Q$. W.l.o.g., assume $|S| \geq 2$. For each state $q \in S$, we add a copy $\hat{q}$ to $Q'$. Further, we add the states $q_<$ and $q_>$. 
		Let $q_1, q_2, \dots, q_n$ be an ordering of the states in $Q$ such that $\delta$ follows this ordering.
		The transition function $\delta'$ agrees with $\delta$ on all states in $Q$ and letters in~$\Sigma$. For a copied state $\hat{q}$, we set $\delta'(\hat{q}, \sigma) = \hat{q}$ for all $\sigma \in \Sigma$ and $\delta'(\hat{q}, c) = q$. For every state $q \in Q$, we set $\delta'(q, c) = q_<$.
		Let $q_s$ be some state in~$S$. Then 
		for all $\sigma \in \Sigma$ we set $\delta'(q_<, \sigma) = \delta(q_s, \sigma)$, $\delta'(q_<, c) = q_s$ and $\delta'(q_>, \sigma) = q_>$, $\delta'(q_>, c) = q_s$.
		Then, we set $R = \{(q_i, q_j)\mid i < j \}$ for all states in~$Q$. Further, for every copied state $\hat{q_k}$ we extend $R$ by the sets: $\{(\hat{q_k}, q_k)\}$, $\{(q_i, \hat{q_k}), (\hat{q_k}, q_j) \mid i < k, k < j\}$, and $\{(\hat{q_i}, \hat{q_k}), (\hat{q_k}, \hat{q_j})\mid i < k < j\}$ for all copied states $\hat{q_i}, \hat{q_j}$. For the states $q_<, q_>$, we add $\{(q_<, q) \mid q\neq q_< \in Q'\}$ and $\{(q, q_>) \mid q\neq q_> \in Q'\}$ to $R$.
		
	Assume, $w\in \Sigma^*$ synchronizes the set $S$ in $A$. W.l.o.g., assume $w \neq \epsilon$. Then, $cw$ synchronizes the automaton $A'$ such that $R \subseteq\, \ordp$ (where position $i=0$ is not taken into account). We have in $A'$ that $Q'.c = S \cup \{q_<\}$. Since the initial configuration is not taken into account all transitions are allowed as the first letter of a synchronizing word and hence $R \subseteq\, \ordp[c]$.
	From now on, no transition which leads backwards in the order is allowed. We constructed $R$ such that for the states in $Q$ all transitions inherited from~$\delta$ are valid at any time. 
	Since $w \in \Sigma^*$, all transitions induced by $w$ are valid for states in $Q$ in $A'$. The only active state outside of $Q$ is $q_<$ which mimics transitions of the active state $q_s$ with the next letter $w[1] \in \Sigma$ and hence $q_s$ and $q_<$ are synchronized in the next step. Note that with any word from $\Sigma^*$ no state outside of $Q$ is reachable from a state in $Q$. Hence, $Q'.cw = S.w$.
	
	For the other direction, assume $w \in \Sigma'^*$ synchronizes $A'$ and $R \subseteq\, \ordp$ (where position $i=0$ is not taken into account). Then, the first letter of $w$ needs to be the letter $c$. Otherwise, the state $q_>$ stays in $q_>$. This state cannot be left later anymore since after the first transition the pair $(q_s, q_>)$ in $R$ is active and forbids a transition out of $q_>$. As the state $q_>$ cannot be reached from any other state we caused an active non-synchronizing trap-state. 
	
	For the letter $c$, we have $Q'.c = S \cup \{q_<\}$. 
	After the first transition for all active states in~$Q$, the transition by letter $c$ is not allowed anymore as it would yield the states to reach the state $q_<$ which is smaller than any state in $Q$. 
	For all letters $\sigma \neq c$, the state $q_<$ simulates the active state $q_s$ and hence synchronizes with it with the letter $w[2]$. Starting from $Q$ we stay in $Q$ with all $\sigma \in \Sigma$ and simulate the automaton $A$.
%
	Hence, any word that synchronizes the set $S \cup \{q_<\}$ in $A'$ also synchronizes the set $S$ in $A$.
	\end{proof}
\begin{figure}	
	\centering
	\scalebox{0.8}{
	\begin{tikzpicture}[->,>=stealth',shorten >=1pt,auto,node distance=2cm,semithick,state/.style={circle, draw, minimum size=.9cm}]

\node[state] (q<) {$q_<$};
\node[state] (q1) [right of=q<] {$q_1$};
\node[state] (q2) [right of=q1] {$q_2$};
\node[state] (q3) [right of=q2] {$q_3$};
\node[state] (q4) [right of=q3] {$q_4$};
\node[] (dot) [right of=q4] {$\cdots$};
\node[state] (qn) [right of=dot] {$q_n$};
\node[state] (q>) [right of=qn] {$q_>$};

\node[state] (q2s) at (4,1.5) {$\hat{q_2}$};
\node[state] (q4s) at (8,1.5) {$\hat{q_4}$};

\path 
(q1) edge node {$c$} (q<)
(q2) edge [bend left] node [pos=0.1, below right] {$c$} (q<)
(q3) edge [bend left] node [pos=0.08, below right] {$c$} (q<)
(q4) edge [bend left] node [pos=0.06, below right] {$c$} (q<)
(qn) edge [bend left] node [pos=0.03, below right] {$c$} (q<)
(q>) edge [bend right=50] node {$c$} (q2)
(q<) edge [bend left=90] node [] {$\Sigma$} (q3)
(q<) edge [bend left=30] node [] {$c$} (q2)
(q>) edge [loop below] node {$\Sigma$} (q>)
(q2s) edge [loop left] node {$\Sigma$} (q2s)
(q4s) edge [loop right] node {$\Sigma$} (q4s)
(q2s) edge node {$c$} (q2)
(q4s) edge node {$c$} (q4)
(q2) edge node [above] {$\Sigma$} (q3)
;
\end{tikzpicture}}
	\caption{Schematic illustration of the reduction from the subset synchronization problem for complete weakly acyclic automata (see Theorem~\ref{thm:to1NPc}). In this example, the subset $S$ contains the states $q_2$ and $q_4$, we picked $q_s = q_2$. Transitions inherited from the original automaton $A$ are not drawn except for the transitions from $q_2$, for illustration they were assumed to lead to $q_3$.}
	\label{fig:syncsubset}
\end{figure}
\begin{theorem}
	\label{thm:subset12}
	{\sc Subset-Sync-Under-Total-$\ordp$} is \NP-complete.
\end{theorem}
%
\begin{proof}
	By Theorem~\ref{O3TotalEquivPWAA} we know that the automata -- which are synchronizable under the constraint formulated in {\sc Sync-Under-Total-$\mathit{0}$-$\ordp$} -- are precisely the synchronizable weakly acyclic automata. Since every complete weakly acyclic automaton (CWAA) is also a PWAA the \NP-hardness of the subset synchronization problem for CWAAs transfers to the problem {\sc Subset-Sync-Under-Total-$\mathit{0}$-$\ordp$} in our setting. Since the problem {\sc Sync-Under-Total-$\mathit{1}$-$\ordp$} is already \NP-hard it follows by setting $S := Q$ that {\sc Subset-Sync-Under-Total-$\mathit{1}$-$\ordp$} is also \NP-hard.
	The length bound obtained in Theorem~\ref{thm:O3totalNP} also holds for a shortest word $w$ synchronizing a subset $S$ with $w\in \ordp$ if $R$ is a strict and total order.
	This gives membership in \NP\ as we can guess the synchronizing word.
\end{proof}
\begin{corollary}
	Let $A = (Q, \Sigma, \delta)$ with $n= |Q|$, $S\subseteq Q$, and $R \subseteq Q^2$ be a strict and total order on $Q$.
	If $S$ is synchronizable in $A$ by a shortest word $w$ such that $R \subseteq\, \ordp$, then $|w|\leq \frac{n(n-1)}{2} +1$ for {\sc Subset-Sync-Under-Total-$\ordp$}.
\end{corollary}
\begin{theorem}
	\label{thm:subset3}
	The following subset synchronization problems are \PSPACE-complete for both -$\mathit{0}$- and -$\mathit{1}$-: {\sc Subset-Sync-Under-$\ore$}, -$\orep$, -$\orz$, -$\orzp$, -$\ordp$.
\end{theorem}
\begin{proof}
		For the mentioned orders, the subset synchronization problem is trivially \PSPACE-hard which can be observed by setting $R=\emptyset$. In order to show membership in \PSPACE, the start states in the considered powerset-construction in Theorem~\ref{thm:inPSPACE} can be adapted to check reachability from the start configuration where exactly the subset $S$ is active to some final state using polynomial~space.~
	\end{proof}
Several other results can be transferred from~\cite{DBLP:journals/tcs/Ryzhikov19a} to the corresponding version of the {\sc Sync-Under-Total-$\mathit{0}$-$\ordp$} problem, such as inapproximability of the problems of finding a shortest synchronizing word; a synchronizing set of maximal size (here also \We-hardness can be observed); or determining the rank of a given set. 
%
Further, by the observation (in~\cite{DBLP:journals/tcs/Ryzhikov19a}) that, in the construction given in~\cite{Rys80,DBLP:journals/siamcomp/Eppstein90} the automata are WAAs, we immediately get \NP-hardness for finding a shortest synchronizing word for all of our orders (for order $l<l$ and $l\leq l$ set $R = \emptyset$).
\begin{corollary}
	For all considered orders $\lessdot_w$, the problem given a DCA $A=(Q, \Sigma, \delta)$, $k\in \mathbb{N}$, $R\in Q^2$, if there exist a synchronizing word $w\in \Sigma^*$ with $|w|\leq k$ and $R\subseteq \lessdot_w$ is \NP-hard.
\end{corollary}
\section{Transferred Results}
\label{sec:transResult}
In \cite{DBLP:journals/tcs/Ryzhikov19a} weakly acyclic automata are considered, which are complete deterministic automata for which the states can be ordered such that no transitions leads to a smaller state in the order. 
We proved in Theorem~\ref{O3TotalEquivPWAA} that the class of automata considered in {\sc Sync-Under-Total-$\mathit{0}$-$\ordp$} is equivalent to the class of partial weakly acyclic automata (PWAA).
In \cite{DBLP:journals/tcs/Ryzhikov19a} the corresponding class of complete weakly acyclic automata is investigated and several hardness results are obtained for different synchronization problems concerning this class of automata. Since complete weakly acyclic automata are a subclass of partial weakly acyclic automata the obtained hardness results easily transfer into our setting. Note that in~\cite{DBLP:journals/tcs/Ryzhikov19a} the approximation results are measured in $n=|Q|$ and not in the size of the input. Hence, the results can be directly transferred despite the fact that in the problem {\sc Sync-Under-Total-$\mathit{0}$-$\ordp$} the input is extended to include the set $R$ of size $|Q|^2$. 
We refer to the decision variant of an optimization problem by the extension \emph{-D} in its name. 
The following results transfer from~\cite{DBLP:journals/tcs/Ryzhikov19a}:
\begin{definition}[\textmd{{\sc Short-Sync-Word-Total-$\mathit{0}$-$\ordp$}}]
	\index{\textsc{Short-Sync-Word-Total-$\mathit{0}$-$\ordp$}}
	Given a DCA $A = (Q, \Sigma, \delta)$, and a strict and total order $R \subseteq Q^2$.
	Output the length of a shortest word $w$ such that $|Q.w| = 1$ and $R \subseteq\, \ordp$.
\end{definition}
\begin{corollary}
	The problem {\sc Short-Sync-Word-Total-$\mathit{0}$-$\ordp$} for $n$-state binary automata cannot be approximated in polynomial time within a factor of $\mathcal{O}(n^{\frac{1}{2} - \epsilon} )$ for any $\epsilon > 0$ unless \PTIME = \NP.
\end{corollary}
\begin{definition}[\textmd{{\sc Max-Sync-Set-Total-$\mathit{0}$-$\ordp$}}]
	\index{\textsc{Max-Sync-Set-Total-$\mathit{0}$-$\ordp$}}
	Given a DCA $A = (Q, \Sigma, \delta)$, and a strict and total order $R \subseteq Q^2$.
	Output a set $S\subseteq Q$ of maximum size such that $|S.w| = 1$ and $R \subseteq\, \ordp$.
\end{definition}
\begin{corollary}
	The problem {\sc Max-Sync-Set-Total-$\mathit{0}$-$\ordp$} for $n$-state automata over an alphabet of cardinality $\mathcal{O}(n)$
	cannot be approximated in polynomial time within a factor of $\mathcal{O}(n^{1-\epsilon})$ for any $\epsilon > 0$ unless \PTIME = \NP.
\end{corollary}
\begin{corollary}
	The problem {\sc Max-Sync-Set-Total-$\mathit{0}$-$\ordp$} for binary $n$-state automata cannot be approximated in polynomial time within a factor of $\mathcal{O}(n^{\frac{1}{3}- \epsilon})$ for any $\epsilon > 0$ unless \PTIME = \NP.
\end{corollary}
Observing the reductions given in~\cite{DBLP:journals/tcs/Ryzhikov19a} to obtain the above transferred inapproximability results, we also conclude the following hardness results concerning the parameterized complexity class \We.
\begin{corollary}
	The problem {\sc Max-Sync-Set-Total-$\mathit{0}$-$\ordp$}-D is \We-hard with the parameter $k$ being the given size bound on the set $S$ in the decision variant of the problem.
\end{corollary}
\begin{definition}[\textmd{{\sc Set-Rank-Total-$\mathit{0}$-$\ordp$}}]
	\index{\textsc{Set-Rank-Total-$\mathit{0}$-$\ordp$}}
	Given a DCA $A = (Q, \Sigma, \delta)$, a subset $S\subseteq Q$ and a strict and total order $R \subseteq Q^2$.
	Output the rank of $S$ in $A$ under $\ordp$, that is the size of a smallest set $S'$ such that there exists a word $w$ with $S.w = S'$ and $R \subseteq\, \ordp$.
\end{definition}
\begin{corollary}
	The problem {\sc Set-Rank-Total-$\mathit{0}$-$\ordp$} for $n$-state  automata with alphabet of size $\mathcal{O}(\sqrt{n})$ cannot be approximated within a factor of $\mathcal{O}(n^{\frac{1}{2}-\epsilon})$ for any $\epsilon > 0$ unless \PTIME = \NP.
\end{corollary}
\begin{corollary}
	The problem {\sc Set-Rank-Total-$\mathit{0}$-$\ordp$} for $n$-state binary automata cannot be approximated within a factor of $\mathcal{O}(n^{\frac{1}{3}-\epsilon})$ for any $\epsilon > 0$ unless \PTIME = \NP.
\end{corollary}
\begin{definition}[\textmd{{\sc Sync-Into-Subset-$\mathit{0}$-$\ordp$}}]
	\index{\textsc{Sync-Into-Subset-$\mathit{0}$-$\ordp$}}
	Given a DCA $A = (Q, \Sigma, \delta)$, a subset $S\subseteq Q$ and a strict and total order $R \subseteq Q^2$.
	Does there exist a word $w$ with $Q.w = S$ and $R \subseteq\, \ordp$.
\end{definition}
\begin{corollary}
	The problem {\sc Sync-Into-Subset-$\mathit{0}$-$\ordp$} is \NP-hard.
\end{corollary}
\section{Conclusion}
We discussed ideas how constraints for the design of assembly lines caused by the physical deformation of a part can be described in terms of synchronization problems. For that, we considered several ways how a word can imply an order of states in $Q$. We considered the complexity of synchronizing an automaton under different variants of orders and observed that the complexity of considering an order on the set of active states may differ from considering the order on each single path. 
Although we were able to get a good understanding of the complexity of synchronization under the considered orders, some questions remained open: We only know that {\sc Sync-Under-$\orep$} is contained in \NP\ but it is open whether the problem is \NP-complete or if it can be solved in polynomial time. Conversely, for {\sc Sync-Under-$\mathit{0}$-$\orzp$} the problem is \NP-hard but its precise complexity is unknown. It would be quite surprising to observe membership in~\NP\ here since it would separate the complexity of this problem from the closely related problem {\sc Sync-Under-$\mathit{1}$-$\orzp$}. 
Further, it remains open whether for the other orders
a drop in the complexity can be observed, when $R$ is strict and total, as it is the case for $\ordp$.
\index{dynamic constraints|)}
\nocite{DBLP:journals/ipl/Rystsov83}
\nocite{DBLP:journals/jcss/Savitch70}
\nocite{DBLP:conf/dagstuhl/Sandberg04}
\bibliography{sync_under_order}
\end{document}